\documentclass[sigconf, 10pt, nonacm]{acmart}
\makeatletter
\def\@ACM@checkaffil{
    \if@ACM@instpresent\else
    \ClassWarningNoLine{\@classname}{No institution present for an affiliation}%
    \fi
    \if@ACM@citypresent\else
    \ClassWarningNoLine{\@classname}{No city present for an affiliation}%
    \fi
    \if@ACM@countrypresent\else
    \ClassWarningNoLine{\@classname}{No country present for an affiliation}%
    \fi
}
\makeatother

\definecolor{opt}{RGB}{198,239,206}      
\definecolor{near}{RGB}{255,235,156}    
\definecolor{mid}{RGB}{255,204,153}     
\definecolor{poor}{RGB}{255,204,229}    
\definecolor{bad}{RGB}{220,140,140}
\definecolor{darkgreen}{RGB}{190,95,0}
\definecolor{lila}{HTML}{FF00FF}
\usepackage{enumitem}
\usepackage[ruled, lined, linesnumbered, commentsnumbered, longend]{algorithm2e}
\usepackage{colortbl}
\usepackage{xspace}

\begin{document}
\title[\textsc{Trivance}: Latency-Optimal AllReduce]{\textsc{Trivance}: Latency-Optimal AllReduce by \\Shortcutting Multiport Networks}

\author{Anton Juerss}
\affiliation{
  \institution{Weizenbaum Institute \& TU Berlin}
}

\author{Vamsi Addanki}
\affiliation{
  \institution{Purdue University}
}

\author{Stefan Schmid}
\affiliation{
  \institution{TU Berlin \& Weizenbaum Institute}
}

\renewcommand{\shortauthors}{Juerss, Addanki, Schmid}

\begin{abstract}
AllReduce is a fundamental collective communication operation in distributed computing and a key performance bottleneck for large-scale training and inference. Its completion time is determined by the number of communication steps, which dominate latency-sensitive workloads, and the communication distance affecting both latency- and bandwidth-bound regimes. Direct-connect topologies, such as Google's TPUv4 tori, are particularly prone to large communication distances due to limited bisection bandwidth.

In this paper, we present \textsc{Trivance}, a novel AllReduce algorithm that completes within $\log_3 n$ steps---a 50\% improvement in comparison to Swing and Recursive Doubling, while reducing congestion compared to Bruck's algorithm by a factor of three and preserving bandwidth-optimality. \textsc{Trivance} exploits both transmission ports of a bidirectional ring within each step to \emph{triple} the communication distance along both directions simultaneously. By performing joint reductions, \textsc{Trivance} improves both the number of steps and network congestion. We further show that \textsc{Trivance} extends naturally to multidimensional torus networks, retaining its latency advantage while achieving performance comparable to bandwidth-optimal algorithms for large AllReduce sizes.

Our packet-level SST simulation shows that \textsc{Trivance} improves state-of-the-art approaches by 5-30\% for AllReduce sizes up to 8\,MiB, in high-bandwidth settings up to 32\,MiB and for 3D tori up to 128\,MiB. Throughout the evaluation, \textsc{Trivance} remains the best-performing latency-optimal algorithm.
\end{abstract}

\maketitle

\section{Introduction}
\label{sec:introduction}
Collective communication lies at the heart of many high-performance computing applications and machine learning (ML), both for training and inference. As ML model sizes continue to grow~\cite{NEURIPS2020_1457c0d6,shoeybi2019megatron,10.1145/3458817.3476209}, efficient collective communication primitives are critical for minimizing training time and maximizing hardware utilization~\cite{10.1145/3437801.3441620,10.1145/3458817.3476209}. Recent research has focused on improving collective communication performance in real networks, spanning general algorithm design~\cite{295653,10.1145/2686882,10.1145/3524059.3532380}, synthesizing custom algorithms for specific topologies~\cite{10.1145/3437801.3441620,285084,10.1145/3651890.3672249}, and building scalable network topologies for large-scale training~\cite{TPUv4,zu2024resiliency}.

AllReduce stands out as a fundamental collective primitive, aggregating and disseminating data such as gradients across nodes during parallel training~\cite{DESENSI202470,JOCKSCH2021102812,295653}. AllReduce often dominates in both HPC and deep learning environments where the latter heavily relies on AllReduce for gradient synchronization. In production HPC workloads, measurements on large-scale systems have shown that more than 40\% of the total Message Passing Interface (MPI) time can be spent in MPI\_AllReduce and MPI\_Reduce combined~\cite{rabenseifner1999automatic, Chunduri2018Supercomputer, 285119}.
\begin{figure*}[t]
  \centering
  \includegraphics[width=0.95\textwidth]{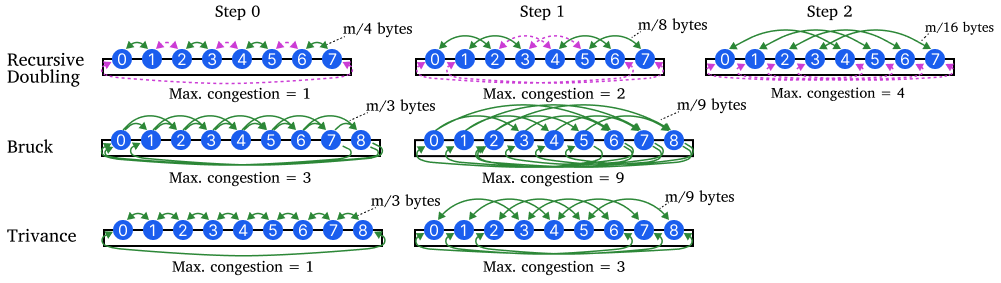}
  \Description{Three rows compare ring communication schedules.
    Recursive Doubling on eight nodes requires three steps, with maximum
    link congestion increasing from one to two to four. Bruck on nine
    nodes requires two steps but routes traffic predominantly in one
    direction, producing maximum congestion of three and nine.
    Trivance also requires two steps, but communicates symmetrically in
    both directions and limits maximum congestion to one and three.
    Green arrows indicate the original collective, dashed purple lines the mirrored one.}
  \caption{Compared to Recursive Doubling requiring $\log_2 n$ steps, \textsc{Trivance} completes in $\log_3 n$ steps---a 50\% improvement---by leveraging both ports for joint reductions and tripled communication distance per step. Congestion is reduced by a factor of three compared to Bruck. Green arrows indicate the original collective, dashed purple lines the mirrored one.}
  \label{fig:trivancevsrecursivedoublingbruck}
\end{figure*}

With the adoption of torus topologies in large-scale GPU clusters~\cite{TPUv4,zu2024resiliency}, optimizing AllReduce for multiport, bidirectional networks is key to minimizing communication overheads~\cite{295653}. In such networks, each node connects to its immediate neighbors via bidirectional links, enabling it to send and receive messages concurrently. The structural simplicity and cost-efficiency of torus topologies make them a compelling choice for large-scale deployments, particularly in the context of machine learning workloads~\cite{10.1145/3623490}. For this reason, several accelerator systems adopt torus-like topologies, most notably Google's TPU platforms: TPUv4 organizes chips into $4\times4\times4$ cubes using fixed electrical ICI links and connects up to 64 such cubes through same bandwidth optical ICI links routed by a reconfigurable circuit-switched fabric, forming 3D tori of up to 4,096 chips in different shapes~\cite{TPUv4}. More recently, Google's TPU 8t scales a single superpod to as many as 9,600 chips~\cite{google2026tpu8t}.

Collective communication typically proceeds in a series of steps, during which designated nodes exchange messages of a certain size, which may vary depending on the algorithm. The main challenge is to fully exploit the available bandwidth while minimizing latency and congestion. This requires carefully balancing the number of communication steps (which impacts startup latency) with the total data exchanged (which affects transmission delay), both of which determine AllReduce completion time. Achieving minimal message congestion overlap and minimal communication steps simultaneously is mutually exclusive.

For bidirectional ring topologies, Bruck establishes the lower bound for completion of both All-to-All and AllGather collectives in $\log_3 n$ communication steps~\cite{Bruck1994}. In contrast, much of the existing work on AllReduce has traditionally focused on single-ported network models, under which the completion is fundamentally bounded by $\log_2 n$ communication steps~\cite{10.1007/978-3-540-24685-5_1}. Despite the assumption of bidirectional links in ring-like topologies, the best-known latency-optimal algorithms for AllReduce continue to require at least $\log_2 n$ communication steps on a network of $n$ nodes~\cite{10.1007/978-3-540-24685-5_1,295653,RUEFENACHT201724}.

Figure~\ref{fig:trivancevsrecursivedoublingbruck} illustrates the difference between traditional approaches: Recursive Doubling~\cite{10.1007/978-3-540-24685-5_1} takes three steps and incurs congestion from four overlapping transmissions in its final step for a network of 8 nodes. In comparison, Bruck's algorithm~\cite{Bruck1994} completes in two communication steps, but induces more link congestion compared to \textsc{Trivance} as it routes all traffic in a single direction. \textsc{Trivance} completes in two steps  with uniform link congestion of one in the first step and three in the second.

In this paper, we demonstrate that latency-optimality and uniform network utilization can be achieved simultaneously. We present \textsc{Trivance}, a new AllReduce algorithm that completes in $\log_3 n$ steps while preserving bandwidth-optimality. The key insight is that in bidirectional topologies such as rings and tori, each node can simultaneously utilize both ports to \emph{triple} the communication distance in each step and jointly reduce both received transmissions. These joint reductions increase the total number of blocks reduced per step and differ from existing solutions that perform reductions independently. As shown in Figure~\ref{fig:trivancevsrecursivedoublingbruck}, this improves the number of steps by approximately $1.5\times$ compared to traditional AllReduce algorithms such as Recursive Doubling, while improving network congestion by a factor of three relative to Bruck. By effectively shortcutting the ring, \textsc{Trivance} completes AllReduce in the same number of communication steps as Bruck, but with lower end-to-end latency and higher uniform utilization of network links. More generally, we show that \textsc{Trivance} significantly reduces total completion time for short and moderate sized messages.

Our evaluations across a wide range of network characteristics and AllReduce sizes demonstrate that \textsc{Trivance} improves the completion time of AllReduce by 5-30\% for AllReduce sizes up to 8\,MiB for 2D tori; for high-bandwidth networks up to 32\,MiB. For 3D tori, \textsc{Trivance} outperforms state-of-the-art algorithms even for AllReduce sizes of 128\,MiB. Throughout the entire evaluation, \textsc{Trivance} remains the best-performing latency-optimal algorithm.

\noindent The key contributions of this work are as follows:
\begin{itemize}[topsep=0pt,partopsep=0pt]
    \item We introduce \textsc{Trivance}, a novel AllReduce algorithm that leverages bidirectional communication to improve the number of required communication steps by 50\% compared to traditional AllReduce approaches while reducing congestion by a factor of three relative to Bruck. \textsc{Trivance} performs joint reductions from both incoming messages per step to complete AllReduce in $\log_3 n$ steps.
    \item We derive theoretical bounds for \textsc{Trivance} and provide novel analytical insights into state-of-the-art algorithms for torus networks with respect to latency-, bandwidth-, and transmission delay optimality. 
    \item We conduct large-scale evaluations using the Structural Simulation Toolkit (SST)~\cite{janssen2010simulator}, comparing existing AllReduce approaches and \textsc{Trivance} across AllReduce sizes from $32\,$B to $128\,$MiB for tori with $D=1,2,3$.
    \item To support reproducibility and facilitate follow-up work, we release our simulation code and evaluation artifacts as open source: \url{https://doi.org/10.5281/zenodo.21180442}. 
\end{itemize}

\emph{This work does not raise any ethical issues.}

\section{Model and Preliminaries}
\label{preliminaries}

Our key insight lies in exploiting the capabilities of multiport topologies more effectively for latency optimization than prior approaches. To build the intuition behind our approach, we introduce a simple model and terminology that guides both our analysis and algorithm design. We begin by considering a ring topology with $n$ nodes, where each node is connected to its immediate neighbors via bidirectional links. Each node is equipped with two ports, one in each direction, enabling it to send and receive two messages concurrently, one per port. We assume that packets are forwarded using deterministic shortest-path routing, and that the network provides no specialized hardware support for accelerating collective operations~\cite{10.1145/1088149.1088183}. Our approach, and hence its guarantees and evaluation, target regular direct-connect topologies, such as rings and multidimensional tori.

\subsection{Latency, Bandwidth \& Congestion}
\label{sec:costfunction}
In order to reason about the performance of our approach, we revisit the classic latency and bandwidth Hockney model~\cite{hockney}. The communication cost to send $m$ bytes from point-to-point is modeled as $C(m) = \alpha + m \cdot \beta$, where $\alpha$ represents the network latency and $\beta$ the bandwidth cost per byte, defined as the inverse of the network bandwidth. The parameters $\alpha$ and $\beta$ are derived as constants obtained from hardware specifications. 
Typically in collectives, a data reduction operation is involved, which is represented by a term $\gamma$ accounting for the aggregation cost. Since our model and the presented state-of-the-art AllReduce algorithms exhibit comparable aggregation costs, this term is omitted~\cite{10.1007/978-3-540-24685-5_1, RUEFENACHT201724},  but discussed in Section~\ref{sec:discussion}.
In our model, multiple messages can share the same physical link, thereby sharing the available bandwidth and causing potential performance degradation. 
Following prior work~\cite{295653,10.1145/2686882}, we adopt a congestion-aware cost model that estimates the completion time of an AllReduce algorithm $A$ for a message of size $m$ as:
\begin{equation*}
C(m, A) = steps(A) \cdot \alpha + \sum_{k=0}^{steps(A)-1} \beta \cdot m_k \cdot c_k
\end{equation*}
where, in each step $k$, the algorithm incurs a startup latency $\alpha$ (e.g., data preparation and propagation delay) and a transmission latency of $\beta \cdot m_k \cdot c_k$. Here, $\beta = \frac{1}{b}$ denotes the transmission time per bit for bandwidth $b$, $m_k$ the chunk size transmitted in step $k$, and $c_k$ the congestion incurred during that step, i.e., number of chunks sharing a link. In this cost function, $\sum_{k=0}^{steps(A)-1} m_{k} \cdot c_{k}$ represents the transmission delay costs of $m$, where the size of sent data in step $k$ is weighted by the congestion incurred in the network.
\subsection{AllReduce as a Collective Operation}
The optimization of collective algorithms, and thus the effective utilization of communication ports in torus topologies, depends strongly on the workload and network parameters, requiring approaches that either prioritize latency or maximize bandwidth efficiency. AllReduce algorithms operate under the assumption that each node $r \in \{0, \dots, n-1\}$ initially holds a unique local data vector $V^r = \langle v_0^r, v_1^r, \dots, v_{n-1}^r \rangle$ of $m$ bytes, which can be partitioned in at least $n$ elements. The objective is to compute, for each index $r \in \{0, \dots, n-1\}$, a global reduction over the set of values $\{v_r^u \mid u \in \{0, \dots, n-1\}\}$. Generally, latency-optimal approaches require each node to send its entire vector to every other node. Upon completion, every node possesses the data of all nodes. For the bandwidth-optimal approaches (Rabenseifner algorithm~\cite{10.1007/978-3-540-24685-5_1}), the responsibility to reduce index $r$ of all initial data vectors is typically assigned to node~$r$. In order for node $r$ to compute the reduced value for its assigned index, it must receive all corresponding entries $v_r^u$ from every node $u$ in the network. Consequently, bandwidth-optimal AllReduce algorithms decompose into two phases: a Reduce-Scatter phase and an AllGather phase.
\subsection{Optimality for $D$-dimensional torus networks for Latency and Bandwidth}
\label{sec:optimality}
The required number of communication steps and transmission delay of collectives are fundamentally constrained by both the number of available communication ports and the underlying network topology. Chan et al.~\cite{Chan2006Multiport} establish a tight lower bound for collective communication primitives for $D$-dimensional tori including $n$ nodes, where each node has two ports per dimension for a total of $2D$ ports. The completion of AllReduce requires at least $\lceil \log_{2D+1} n \rceil$ communication steps by constructing a minimal spanning tree over the $2D$ incident links. For bidirectional rings, this translates to a lower bound of $\lceil\log_3 n\rceil$. With respect to transmission delay, referred to as bandwidth cost by Chan et al.~\cite{Chan2006Multiport}, the lower bound is given by $2\frac{n-1}{n} \cdot \frac{m\beta}{2D}$, which is commonly approximated as $\frac{m\beta}{D}$. Achieving this bound requires that each data block traverses every link at most once, ensuring congestion free transmission without overlapping routes. Table~\ref{tab:optimality} summarizes, for each algorithm, the relative factors with respect to the optimal latency $\Lambda$ (number of steps), bandwidth $\Delta$ (transmitted data per node), and transmission delay $\Theta$ on a bidirectional ring topology. Transmission delay optimality corresponds to the fraction of the data vector~$m$ that is transmitted, multiplied by the network congestion. Bruck and \textsc{Trivance} achieve the bidirectional ring latency lower bound of $\log_3 n$ steps. All bandwidth-optimal variants achieve optimal bandwidth cost, but only Ring is transmission delay optimal.
\begin{table}[t]
\renewcommand{\arraystretch}{1.4} 
\centering
\begin{tabular}{l|c|c|c}
\textbf{Algorithm} & \textbf{La. Opt.($\Lambda$)} & \textbf{Bw. Opt.($\Delta$)} & \textbf{Tx. Opt.($\Theta$)} \\
\hline
Ring & \cellcolor{bad}$\frac{2n}{\log_3 n}$& \cellcolor{opt}1 & \cellcolor{opt}1 \\
Rec.Doub. (B)& \cellcolor{poor}$2 \cdot\log_2 3$ & \cellcolor{opt}1 & \cellcolor{near}$\frac{1}{2} \log_2 n$ \\
Swing (B) & \cellcolor{poor}$2 \cdot\log_2 3$ & \cellcolor{opt}1 & \cellcolor{near}$\frac{1}{3} \log_2 n$ \\
Bruck (B) & \cellcolor{mid}$2$ & \cellcolor{opt}1 & \cellcolor{mid}$2\log_3 n$ \\
\textsc{Trivance} (B) & \cellcolor{mid}$2$ & \cellcolor{opt}1 & \cellcolor{near}$\frac{2}{3} \log_3 n$ \\
\hline
Rec.Doub. (L)& \cellcolor{near}$\log_2 3$ & \cellcolor{mid}$\frac{\log_2 n}{2}$ & \cellcolor{bad}$n$ \\
Swing (L) & \cellcolor{near}$\log_2 3$ & \cellcolor{mid}$\frac{\log_2 n}{2}$ & \cellcolor{poor}$\frac{n}{3}$  \\
Bruck (L) & \cellcolor{opt}$1$ & \cellcolor{mid}$\log_3 n$ & \cellcolor{bad}$\frac{3n}{2}$  \\
\textsc{Trivance} (L) & \cellcolor{opt}$1$ & \cellcolor{mid}$\log_3 n$ & \cellcolor{poor}$\frac{n}{2}$ \\
\end{tabular}
\caption{Algorithm optimalities for rings. (L) and (B) denote latency-optimal and bandwidth-optimal variants, respectively. All algorithms use two ports. ``Tx.`` indicates transmission delay, ``La.`` Latency and ``Bw.`` Bandwidth. Latency, bandwidth and Tx delay optimality are with respect to optimal value of $\log_3 n\ $, $2m\ $, $m\beta$. Cell colors denote: green = optimal; yellow/orange = increasing overhead; red = high overhead. }
\label{tab:optimality}
\vspace{-10px}
\end{table}

\noindent Both latency and transmission delay lower bounds decrease with the number of available ports, so for a fixed network size, higher-dimensional torus topologies can achieve lower latency and transmission delay. Moreover, the number of data blocks that must be reduced per step is tightly coupled to the total number of communication steps: fewer steps require joint aggregation of incoming blocks for the same collective. Algorithms such as Bruck and \textsc{Trivance}, which complete in $\log_3 n$ steps, must reduce $2\cdot 3^k$ blocks per step $k$ (i.e., $3^k$ per port) in order to complete within the latency-optimal bound. As a result, each step involves the transmission of larger data volumes compared to Swing or Recursive Doubling over progressively longer distances, which can increase link congestion resulting in higher transmission delay.
\subsection{State-of-the-art Algorithms}
\label{sec:algorithms}
In the following, we briefly review state-of-the-art AllReduce algorithms developed for torus network topologies. These algorithms later serve as the baseline for evaluating the performance of \textsc{Trivance}.
\paragraph{Hamiltonian Rings and Bucket}
The algorithm operates in two phases~\cite{Patarasuk2009AllReduce}: Reduce-Scatter followed by AllGather. In the Reduce-Scatter phase, each node splits its vector of length $m$ into $n$ equal blocks and, over $n-1$ steps, sends one block to its right neighbor and receives one from its left, reducing the incoming block with its local data. In the following AllGather phase, nodes propagate their reduced block in the reverse direction, again over $n-1$ steps, until all nodes hold the complete reduced vector. Each node sends a total of $2m$ bytes requiring $2(n-1)$ communication steps. Consequently, Hamiltonian Rings are bandwidth and transmission delay optimal ($\Delta = \Theta = 1$) with a latency-optimality factor of $\Lambda = \frac{2n}{\log_3 n}$ for rings and tori.

Since the Bucket algorithm operates identically on a ring and outperforms Hamiltonian rings on tori, we use Bucket for our evaluation on multidimensional torus networks~\cite{295653}. For an $a \times a$ two-dimensional torus with $a \cdot a = n$, Bucket first performs a ring Reduce-Scatter along the rows, then a ring Reduce-Scatter along the columns on the partially reduced data, followed by an AllGather phase in reverse order ($\Lambda = \frac{2D\sqrt[D]{n}}{\log_3 n}$). On a $D$-dimensional torus with $2D$ ports per node, Bucket generalizes this approach by performing $2D$ Reduce-Scatter phases followed by $2D$ AllGather phases. Each Reduce-Scatter or AllGather is mapped to a distinct port in a different direction per step, which ensures that no communications from two collectives overlap.

\paragraph{Recursive Doubling.}
The Recursive Doubling AllReduce algorithm distinguishes between a latency- and bandwidth-optimal variant~\cite{10.1007/978-3-540-24685-5_1}. In the latency-optimal version, each node transmits and receives over $\log_2 n$ steps its entire data vector to and from node $q = r$ XOR $2^k$ for $k \in \{0,\dots,\log_2 n -1\}$. By that, each node transmits $\log_2 n \cdot m$ bytes. The latency overhead of this pattern is $\Lambda = \log_2 3 \approx 1.58$ with $\Delta = \frac{\log_2 n}{2}$ and $\Theta = n$. The bandwidth-optimal version (known as Rabenseifner algorithm~\cite{10.1007/978-3-540-24685-5_1}) runs a Reduce-Scatter followed by an AllGather. Instead of transmitting the entire data vector per step, each node divides its data into $n$ blocks. In each step, the size of transmitted data halves while the distance doubles. Each node transmits in total $2m$ bytes within $2\cdot \log_2 n$ steps achieving bandwidth-optimality ($\Delta = 1$) and latency-optimality of factor $\Lambda = 3.16$ with $\Theta = \frac{1}{2}\log n$.

For multidimensional torus networks, Recursive Doubling can be extended by partitioning the data vector $m$ of each node into $2D$ segments and performing $2D$ collectives concurrently.
\paragraph{Swing}
A recently published AllReduce algorithm, so-called Swing~\cite{295653}, is optimized for ring networks by alternating communication directions to balance link utilization. At step $k$, each node $r$ communicates with a partner $\pi(r,k)$:
\begin{center}
    $\pi(r,k) = 
\begin{cases}
r + \rho(k)  \bmod n, \text{ if $r$ is even}\\
r - \rho(k)  \bmod n, \text{ if $r$ is odd}
\end{cases}
$
\end{center}
where $\rho(k) = \sum_{i=0}^k (-2)^i = \frac{1-(-2)^{k+1}}{3}$. Compared to Recursive Doubling, Swing preserves the logarithmic latency of $\log_2 n$ steps while reducing congestion by shortcutting communication 
towards more distant peers. The latency-optimal variant completes in $\log_2 n$ steps ($\Lambda = 1.58$) with $\Delta = \frac{\log_2 n}{2}$, while the bandwidth-optimal version transfers minimal data ($\Delta = 1$) in $2\cdot\log_2 n$ steps ($\Lambda = 3.16$). Due to improved routing compared to Recursive Doubling, their latency-optimal version achieves $\Theta = \frac{n}{3}$ and the bandwidth-optimal version $\Theta = \frac{1}{3}\log_2 n$. Swing also runs for each available port in multidimensional torus networks a separate collective, splitting the data vector by $2D$.
\paragraph{Bruck}
Bruck's concatenation algorithm~\cite{Bruck1994} completes the AllGather collective for $n$ nodes in $s = \lceil \log_{k+1} n \rceil$ communication steps. For ease of exposition, we first describe the communication pattern of Bruck's algorithm on bidirectional ring topologies with two communication ports per node. Each node communicates with two peers per step $k$ of distance $3^k \mod$ $n$ and $2\cdot3^k \mod$ $n$. When we extend Bruck's collective algorithm to the typical AllReduce pattern, the latency-optimal version propagates the entire data vector $m$ for each communication which yields latency-optimality ($\Lambda=1$) with bandwidth-optimality factor of $\Delta = \log_3 n$ and $\Theta = \frac{3n}{2}$.

The bandwidth-optimal version of Bruck is structured similarly to Recursive Doubling, where in both phases $2m$ bytes are propagated per node ($\Delta = 1$) within $2 \log_3 n$ steps ($\Lambda = 2$). Due to all communications going in the same direction, the transmission delay optimality factor is $\Theta = 2\log_3 n$. Although Bruck's algorithm was not originally designed for $D$-dimensional tori, it can be naturally extended by performing $D$ concurrent collectives. In this extension, each collective operates along a different dimension, utilizing the two communication ports associated with that dimension.

\paragraph{Bidirectional Design of Swing, Recursive Doubling and Bucket}
The communication pattern of the Recursive Doubling, Bucket and Swing algorithm initially only utilizes one port per node. To use the additional available port in 1D tori, the bandwidth-optimal versions employ a mirrored or so-called sub-collective in the other direction of the ring. This means the network runs two concurrent but independent collectives with half data each in opposite directions. This generalizes to multidimensional tori: in each dimension, an independent collective is run on a disjoint subset of data. This reduces transmission delay per step, but the total number of communication steps remains unchanged. 
 \section{Motivation}
\label{sec:motivation}
In existing literature, the node's network interface is typically assumed to be single-ported, which implies that each node can send and receive only one message simultaneously. Under this assumption, the Recursive Doubling~\cite{10.1007/978-3-540-24685-5_1} algorithm has been developed to perform the AllReduce operation with minimal number of communication steps, $\log_2 n$. While AllReduce collectives run on a ring topology, each node is physically connected to its two neighbors via dedicated bidirectional links. As we show in the following, this represents a missed opportunity to exploit multiport topologies to reduce the total transmission time in latency-bound scenarios.

\subsection{Problems of Existing Approaches}
Existing research on collective algorithms in ring topologies predominantly assumes that each node exchanges data with only one peer node per communication step. In this domain, Hamiltonian Rings and the Bucket algorithm are effective at reducing large AllReduce sizes, whereas Recursive Doubling and Swing are primarily designed to minimize per step communication latency.
All these approaches extend their model to utilize both available ports per node to reduce data size per communication step. These mirroring algorithms leverage the additional port to replicate the communication pattern in the opposite direction, thereby splitting the transmitted data across both collectives, reducing the message size by $50\%$. However, the number of communication steps remains unchanged as each collective operation is treated independently. Figure~\ref{fig:trivancevsrecursivedoublingbruck} illustrates the original communication pattern of Recursive Doubling using green arrows, while the dashed purple arrows denote the mirrored collective. Both collectives proceed entirely independently, each reducing only the data associated with its own collective. Swing~\cite{295653} behaves structurally similar and claims their algorithm is latency-optimal for bidirectional rings while completing the operation in $\log_2 n$ steps.

The concatenation algorithm proposed by Bruck et al.~\cite{Bruck1994} achieves the theoretical lower bound of $\log_3 n$ communication steps. However, this algorithm is not explicitly designed for any particular physical network topology. Consequently, when operating on a ring, Bruck routes all traffic in a single direction, effectively neglecting the available bandwidth of the reverse links in bidirectional networks, significantly increasing congestion.

We see an opportunity to achieve the lower bound of $\log_3 n$ steps while substantially reducing network congestion. Specifically, if each node can communicate with \emph{two unseen} nodes per step, one in each direction, and perform reduction on both received sets of blocks before forwarding them concurrently through both network directions in the subsequent step, AllReduce can be completed within the latency-optimal bound of $\log_3 n$ with $3\times$ reduced overlapping communications compared to Bruck while preserving latency-optimality.
\subsection{Exploration of Design Space}
Recursive Doubling is the classic latency-optimal algorithm that minimizes the number of communication steps required to complete AllReduce using only a single port per node. For bidirectional links, Bruck's concatenation algorithm completes AllReduce within the minimal required number of steps of $\log_3 n$, but restricts all data transfers to a single direction, leaving the reverse links unused. Our approach builds upon these communication patterns by introducing simultaneous communication in both directions at each step. This extension enables a new latency-optimal algorithm for networks with two ports per node, completing AllReduce on a ring in $\log_3 n$ steps while evenly utilizing network resources. In contrast to Recursive Doubling, our method triples the communication distance in each step to communicate with nodes with entirely unseen data. Compared to Bruck's approach, our algorithm exploits both directions of the ring concurrently to shortcut communication paths, thereby substantially reducing end-to-end latency and significantly mitigating network congestion.

In each step, all nodes reduce incoming data from both communications and will forward the result to both outgoing communication in the next steps. Since each node communicates through both ports simultaneously, network bandwidth is more evenly utilized, which lowers the cost of the most congested message. As a result, the algorithm reaches all nodes in the network in $\lceil \log_3 n \rceil$ steps.

In particular, we aim to design a new latency-optimal algorithm, as well as a novel bandwidth-optimal variant split into a Reduce-Scatter and AllGather phase, which completes the AllReduce collective in $2 \cdot \lceil \log_3 n \rceil$ steps.
\section{Trivance}
\label{sec:trivance}
In this section, we elaborate on the design of the \textsc{Trivance} algorithm for performing the AllReduce collective on the ring topology. The fundamental idea of \textsc{Trivance} is to transfer data simultaneously along both directions of the ring, thereby reaching unseen nodes along both ring directions with the same distance. Consequently, \textsc{Trivance} achieves the following properties:
\begin{itemize}
  \item Completes the AllReduce collective in $\lceil \log_3 n \rceil$ communication steps (Theorem~\ref{the:steps})
  \item Transmits minimal data, ensuring bandwidth-optimality (Lemma~\ref{lem:blocks})
  \item Improves network congestion by $3\times$ to other latency-optimal approaches by shortcutting the ring
  \item Full and uniform utilization of available bandwidth
\end{itemize}
\subsection{Bandwidth-Optimal Algorithm} In the bandwidth-optimal version of \textsc{Trivance}, similar to other bandwidth-optimal algorithms, a Reduce-Scatter is run followed by an AllGather phase~\cite{10.1007/978-3-540-24685-5_1, 295653}. Per step $k$, each node $r$ communicates in a network of $n$ nodes to two peers ($\pi_{\text{left}},\pi_{\text{right}}$) as follows:
\begin{center}
    $\pi(r, k, n) = (\pi_{\text{left}},\pi_{\text{right}}) = 
\begin{cases}
\pi_{\text{left}} = r - \rho(k) \bmod n,\\
\pi_{\text{right}} = r + \rho(k) \bmod n,
\end{cases}
$
\end{center}
The distance of communicating nodes is defined by $\rho(k) = 3^k$. Since all nodes communicate in both directions with the same distance, the congestion on all links across the network is uniform and equals $3^k$ ($\Theta = \frac{2\log_3 n}{3}$). This selection of peers results in a communication pattern shown in Figure~\ref{fig:trivancevsrecursivedoublingbruck}, where in the first step each node communicates in both directions with distance one, in the second step with distance three and so on. This pattern enables the algorithm to reach all nodes in the network with fewer communication hops by effectively shortcutting the ring.

During the Reduce-Scatter phase, each node's local data vector is divided into $n$ blocks ${b_0, \dots, b_{n-1}}$, each of size $\frac{m}{n}$. The total reduction operation is split across all nodes in the network, so that each node $r$ is responsible for computing a partial reduction over all $r$-th blocks $b_r$ collected from every node. Therefore, node $r$ performs $\log_3 n$ steps, communicating concurrently with $\pi_{\text{left}}$ and $\pi_{\text{right}}$.
As node $r$ communicates with its peers, it transmits a specific subset of data blocks rather than its entire data vector. To reach all $n - 1$ other nodes, in each step when node $r$ communicates to node $p$ according to $\pi(r, k,n)$, $r$ must send not only the data intended for $p$ but also the data blocks designated for the nodes that $p$ will reach in the subsequent steps $k+1, \dots, s-1$. Consequently, the cumulative size of this data at step $k$ amounts to $m \cdot \frac{1}{3^{k+1}}$. This means, as the communication distance is tripled each step, the size of sent data by $r$ is divided by a factor of three. After each step, the received data is aggregated at each node over all blocks. To compute the total transmitted data by node $r$ across all steps $0,\dots,s-1$ of the Reduce-Scatter phase, we sum over each message size to both peers.
\begin{lemma}[Bandwidth-optimality]
\label{lem:blocks}
The \textsc{Trivance} algorithm achieves bandwidth-optimality by distributing $m \cdot (n - 1)$ unique bytes in $\log_3 n$ steps per Reduce-Scatter or AllGather phase, $m$ bytes per node.
\[
\sum_{k=0}^{s-1} m \cdot \frac{1}{3^{k+1}}\cdot 2 = m \cdot \frac{\tfrac{1}{3}\left(1 - \left(\tfrac{1}{3}\right)^s\right)}{1 - \tfrac{1}{3}}\cdot 2 = m \left(1 - 3^{-s}\right) = m\left(1 - \frac{1}{n}\right)
\]
\end{lemma}
These properties satisfy the conditions for a bandwidth-optimal algorithm achieving minimal bandwidth cost $(\Delta = 1)$ while completing in $2 \cdot \log_3 n$ steps ($\Lambda=2$). In the subsequent AllGather phase, the reduced blocks are broadcasted across the network, but in reverse order, tripling the data size each step and reducing the communication distance by a factor of three. By the end of this phase, every node holds the complete global reduction result.
\subsection{Latency-Optimal Algorithm} By following the same communication pattern as the bandwidth-optimal version, with the modification of running one phase instead of splitting it into a Reduce-Scatter and AllGather, the latency-optimal \textsc{Trivance} algorithm completes the AllReduce collective in $\log_3 n$ steps ($\Lambda = 1$). For each step, all nodes forward their entire vector to their peers resulting in $m \cdot \log_3 n$ transmitted bytes per node ($\Delta = \log_3 n$) maintaining the same congestion of $3^k$ ($\Theta = \frac{n}{2}$).
\paragraph{Proof of Latency-Optimality for \textsc{Trivance}}
Previously, we introduced the design and performance characteristics of the \textsc{Trivance} algorithm for both its latency- and bandwidth-optimal versions. Existing literature~\cite{Bruck1994, Chan2006Multiport} suggests that the lower bound of required communication steps for two port networks equals $\lceil\log_3 n\rceil$. Therefore, the optimal network sizes for \textsc{Trivance} are powers of three. 

\begin{figure}[b]
    \centering
    \includegraphics[width=0.9\columnwidth]{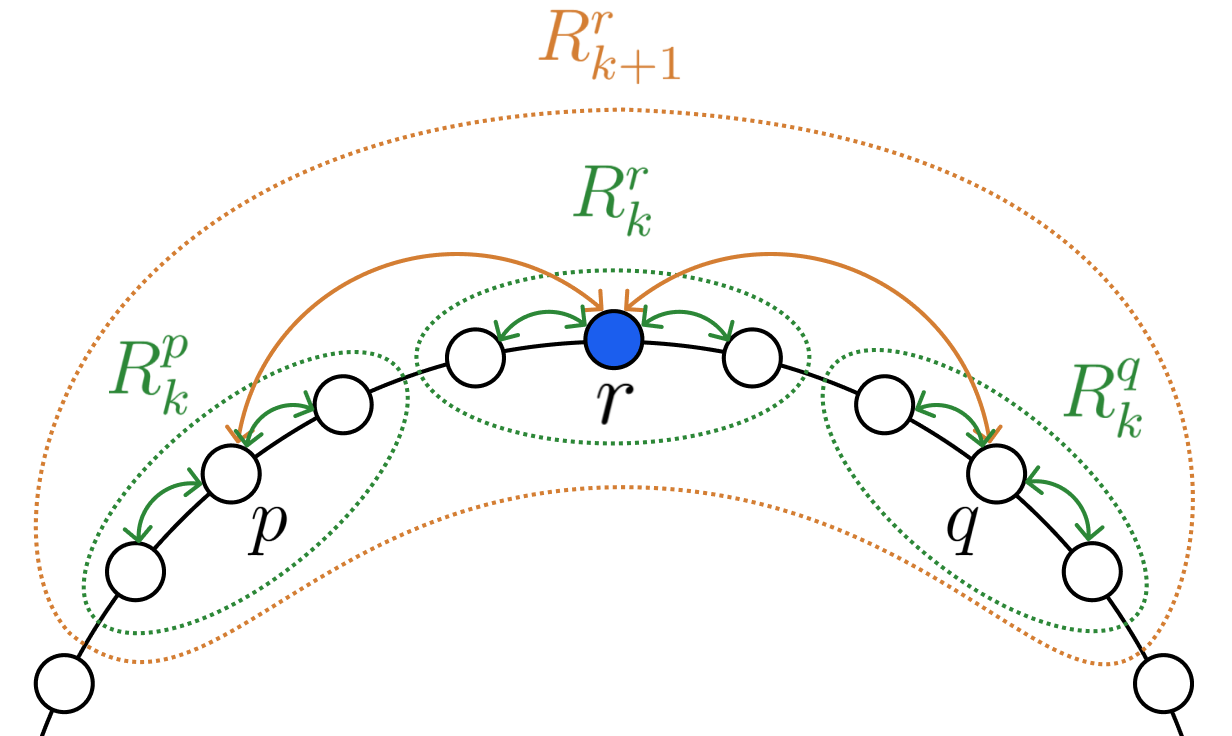}
    \Description{A ring neighborhood centered at node r illustrates block
    propagation between two consecutive steps. At step k, peers p and q
    hold the adjacent left and right ranges, shown by three green dotted
    regions. The larger orange dotted region for step k plus one spans
    the union of these ranges, showing that the data known by r triples
    after receiving from p and q.}
  \caption{Progression of received blocks at node $r$ from step $k$ to step $k+1$, acquiring the blocks of nodes $p$ and $q$ along with their neighbors, tripling the known data.
}
  \label{fig:trivancerange}
\end{figure}

In the following, we sketch our proof that \textsc{Trivance} achieves the latency-optimal lower bound for network sizes of $3^s$ with $s \geq 1$; full proof can be found in Appendix~\ref{A:fullproof}. Its communication pattern is identical to the Reduce-Scatter phase of the bandwidth-optimal variant, but the transmitted data differs. To establish latency-optimality, it suffices to show that each node $r$ receives all data vectors from the other $n-1$ nodes in the network to perform its reduction operation.

This pattern is illustrated in Figure~\ref{fig:trivancerange}, where node $r$ initially receives blocks only from its direct neighbors. In step $k+1$, it receives blocks not only from peers $p$ and $q$, but also from their adjacent neighbors, which forwarded their blocks in step $k$. As a result, the number of distinct blocks known to $r$ tripled from step $k$ to $k+1$. To prove latency-optimality, we first establish Lemma~\ref{lem:range}, which characterizes the set of nodes $r$ has obtained data from after step~$k$. By induction, we show that after step $k$, each node has accumulated all data originating from a continuous neighborhood whose radius increases as $\sum_{i=0}^k 3^i$. 
\begin{lemma}[Block propagation]
\label{lem:range}
After step $k$, each node $r$ holds data originating from nodes within distance $R_k$, defined as:
\[
R_k = \sum_{i=0}^{k} 3^i = \frac{3^{k+1} - 1}{2}
\]
\end{lemma}
\noindent In step $k+1$, communication with peers at distance $3^{k+1}$ results in entirely previously unseen data, since each peer possesses exactly the complementary, disjoint subset required to contribute the blocks between them. As formalized in Theorem~\ref{the:steps}, exchanging these disjoint data sets therefore increases the covered range by $3^{k+1}$, ensuring complete coverage of a network of size $n$ after $\lceil \log_3 n \rceil$ steps.
\begin{theorem}[Latency-optimality]
    \label{the:steps}
    Given a network of size $n = 3^s$, \textsc{Trivance} communication pattern requires $s = \log_3 n$ steps for each node $r$ to receive data from all other $n-1$ nodes in the network.
\end{theorem}
\noindent This allows \textsc{Trivance} to complete AllReduce and AllGather in $\log_3 n$ communication steps, matching Bruck's algorithm and achieving a $\log_2 3$ step reduction compared to traditional latency-optimal AllReduce algorithms such as Swing and Recursive Doubling.

\subsection{Communication Pattern} 
\begin{algorithm}[t]
\small
    \SetKwFunction{peers}{\textbf{\textsc{$\pi$}}}
    \SetKwFunction{blockpropagation}{\textbf{\textsc{BlockPropagation}}}

    \SetKwProg{Fn}{function}{:}{}
    \SetKwInOut{KwIn}{Input}
    \SetKwInOut{KwOut}{Output}

    \KwIn{ $r$: current rank, $step$: current step, $n$: total nodes, $blocks$: block array}
    \KwOut{ Updates $blocks$ to mark all reachable nodes from $r$}

    \Fn{\blockpropagation{$r, step, n, blocks$}}{

        \If{$step \geq \lceil \log_3 n \rceil$}{
            \Return;
        }

        \For{$k \gets step$ \KwTo $\lceil \log_3 n \rceil - 1$}{
            $(peer_{\text{left}}, peer_{\text{right}}) \gets$ \peers{$r, k, n$}\;

            $blocks[peer_{\text{left}}] \gets 1$\;
            $blocks[peer_{\text{right}}] \gets 1$\;

            \blockpropagation{$peer_{\text{left}}, k+1, n, blocks$}\;
            \blockpropagation{$peer_{\text{right}}, k+1, n, blocks$}\;
        }
    }

    \caption{\textsc{Trivance}: Block Propagation}
    \label{alg:get_blocks}
\end{algorithm}
To understand the behavior of the \textsc{Trivance} Reduce-Scatter algorithm, we describe its pattern from the perspective of a node $r$ in a ring network. The AllGather phase functions analogously. For a node $r$, at each step $k$, the communication peers are determined. Furthermore, for each of $r$'s peers, the Block Propagation Algorithm~\ref{alg:get_blocks} recursively computes the set of nodes reached by communication from $r$ to $peer_{\text{left}}$ and $peer_{\text{right}}$, respectively. We follow the recursive construction similar to Recursive Doubling, which has also been the basis for the recent Swing algorithm~\cite{295653}. This process effectively builds a forward reachability map, capturing the propagation paths required to ensure correct block forwarding under \textsc{Trivance}'s recursive communication pattern. 
\begin{figure}[t]
  \centering
  \includegraphics[width=\columnwidth]{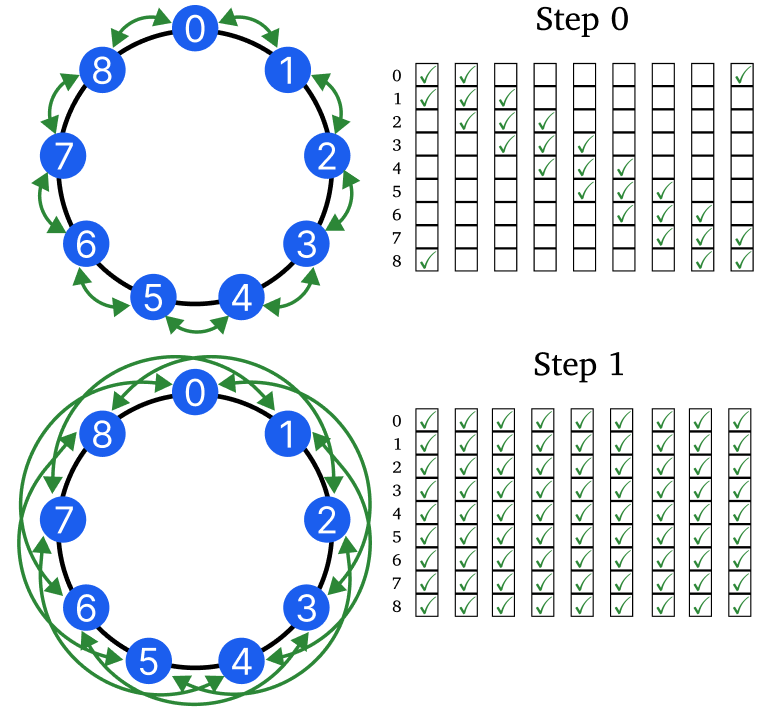}
    \Description{Two panels illustrate block propagation on a nine-node
    ring. In step zero, every node communicates with its immediate
    neighbors, and each row of the adjacent matrix contains check marks
    for only the node itself and its two neighbors. In step one, nodes
    communicate at distance three, and every entry of the matrix is
    checked, indicating that all nodes have received all blocks.}
  \caption{Progression of received blocks for each node on a ring of size $n = 9$.}
  \label{fig:trivance_example}
\end{figure}
For these reached nodes, $r$ sends the corresponding blocks to each peer in step $k$. 

The set of node indices whose blocks must be transmitted from $r$ to $p$ in step $k$ can also be determined using a closed-form expression that summarizes the recursive structure of Algorithm~\ref{alg:get_blocks}:
\[
\Pi(n,r,p,k) = \left\{\, p + \sum_{i=k+1}^{s-1} \epsilon_i \cdot 3^i \;\bmod n \;\middle|\; \epsilon_i \in \{-1, 0, 1\} \right\}
\]
For a node $r$ with $s = \log_3 n$, $\Pi$ computes all paths of subsequent steps where $\epsilon_i$ indicates whether a block is forwarded left (-1), right (1), or not (0) at step $k$. The sum of transmitted blocks in each step accounts to $3^{s-1-k}$, resulting in a total message size of $m \cdot \frac{1}{3^{k+1}}$. It is important to note that, per step $k$, each node must receive and process both incoming requests before it can initiate data transmission at step $k+1$. This condition is necessary because each node must perform a reduction on the data received from both inputs to generate the correct and complete data set to forward in the next step.

Figure~\ref{fig:trivance_example} illustrates the communication pattern and block propagation of \textsc{Trivance} on a ring of size $n=9$, where each node exchanges data bidirectionally with increasing distances over two steps. The vectors below each node of the ring indicate the nodes from which it has received data. For the purpose of this example, we assume that the performed algorithm is the latency-optimal variant of \textsc{Trivance}, wherein each node transmits its complete data vector during each communication round, but the example remains representative considering a Reduce-Scatter collective. Initially, each node holds only its own block, and in step $k = 0$ exchanges the respective data with its immediate neighbors, tripling the number of received blocks. For example, node $0$ collects blocks from nodes $1$ and $8$ in step 0. In step 1, each node communicates now in both directions with distance 3. Notably, the communication partners of node $r$ have already obtained data which $r$ has not received yet. Then, node $0$ has acquired blocks from nodes $1$ and $8$ in step $0$ and now receives the blocks from nodes $2$, $3$, and $4$ via node $3$, and blocks from nodes $5$, $6$, and $7$ via node $6$. As a result, upon completion of step $1$, every node has successfully collected all blocks from all nodes in the network and is prepared to perform the final global reduction operation.
\subsection{Generalizing to Arbitrary Network Sizes}
For ring sizes not equal to a power-of-three, \textsc{Trivance} proceeds as defined for the first $\lfloor \log_3 n \rfloor$ steps. In the final step, each node $r$ communicates in both directions with a distance of $\left\lceil \frac{n - 3^{\lfloor \log_3 n \rfloor}}{2} \right\rceil$ to obtain the remaining blocks required for the reduction. Depending on where the non-power-of-three ring size $n$ places between two optimal power-of-three sizes $3^{\lfloor \log_3 n \rfloor}$ and $3^{\lceil \log_3 n \rceil}$, the distance increases by one for each two nodes of $n$ exceeding $3^{\lfloor \log_3 n \rfloor}$. For example, given a node $r$ in a ring network with 32 nodes, it will have already acquired 27 of the 32 required blocks before the last step. Consequently, only 5 additional blocks need to be forwarded to $r$ from both directions, as opposed to 54 blocks in a network of size 81 (next largest power-of-three size). Figure~\ref{fig:trivanceNonPowerThree} illustrates the communication pattern of \textsc{Trivance} for a network of size 7 compared to the power-of-three network of size 9. As expected, the algorithm completes in two steps for $9$ nodes. For $7$ nodes, two steps are also required despite the smaller network size. The communication distance in the final step is only two (shorter than in the $9$-node case) which leads to fewer communication collisions. Additionally, fewer blocks are transferred in the previous step.

The congestion of the final step is uniform and equals the communication distance as every node communicates in both directions with the same distance. Since the number of missing blocks a node at distance $d$ provides to $r$ is exactly $d$, the block with the reduced information can be obtained by communicating with nearby nodes.
\begin{figure}[t]
    \centering
    \includegraphics[width=\columnwidth]{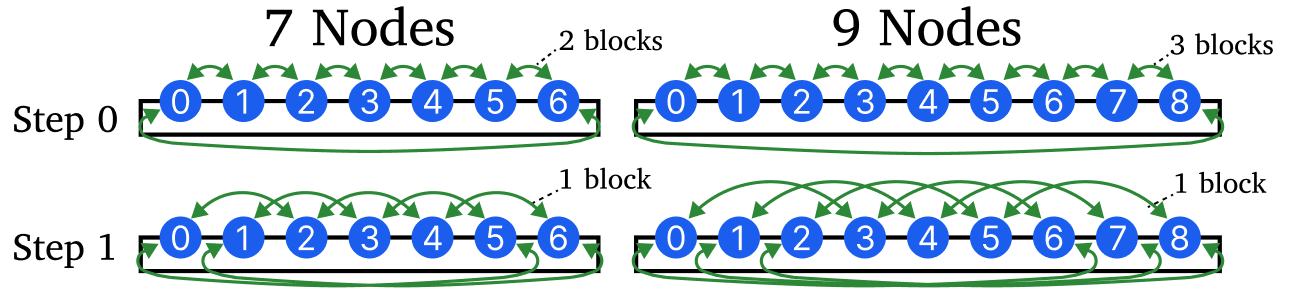}
    \Description{Side-by-side two-step Trivance schedules compare rings
    with seven and nine nodes. In the first step, the seven-node schedule
    transmits two blocks in each direction, whereas the nine-node
    schedule transmits three. In the final step, both transmit one block,
    but the seven-node communication distance is two rather than three.
    Both networks complete in two steps.}
  \caption{\textsc{Trivance} for 7 (left) and 9 (right) nodes.}
  \label{fig:trivanceNonPowerThree}
  \vspace{-15px}
\end{figure}
Consequently, \textsc{Trivance} requires $\lceil \log_3 n \rceil$ communication steps, matching the theoretical lower bound for arbitrary $n$ from Section~\ref{sec:optimality}. Compared to Swing and Recursive Doubling, which require $\lceil \log_2 n \rceil$ steps, the asymptotic improvement in communication steps is $58.5\%$. For power-of-two network sizes, this ranges from $25\%$ to $57\%$, while for arbitrary network sizes it can reach up to $75\%$. Blocks are propagated according to Algorithm~\ref{alg:get_blocks}, ensuring that each block is forwarded at most once per phase and that the total transmitted volume remains bounded by $2m$ per node for AllReduce. Furthermore, the transmission delay cost is strictly lower than that of the next larger power-of-three topology, since the final communication step covers only the remaining uncovered nodes. The exact overhead therefore depends on the position of $n$ between two consecutive powers of three. 

\section{Multidimensional Tori}
\label{sec:multitori}
In Section~\ref{sec:trivance}, we introduced the novel bidirectional communication pattern \textsc{Trivance} for the ring topology. This concept naturally generalizes to higher-dimensional torus networks. Similar to existing approaches~\cite{295653, 10.1145/1810085.1810093, 10.1145/2686882} \textsc{Trivance} performs one collective for each of the $D$ dimensions with $\frac{1}{D}$ of the initial data vector of size $m$.

Figure~\ref{fig:trivanceexamplemultid} illustrates the communication pattern of \textsc{Trivance} on a square torus of size $n=81$. In step 0, the original and mirrored collectives communicate at distance one along the horizontal and vertical dimensions, respectively. They switch dimensions in step 1 and return to their initial dimensions in step 2. Each collective is mapped to a distinct dimension in every step, ensuring that their communications remain disjoint and do not interfere. This allows the transmission of each node's initial data vector of size $m$ to be split across the additional communication links, reducing the message size per transmission. At the same time, nodes communicate with closer neighbors, which decreases both communication distance and congestion. This applies to both the latency- and bandwidth-optimal version of \textsc{Trivance}.

\begin{figure}[t]
    \centering
    \includegraphics[width=\columnwidth]{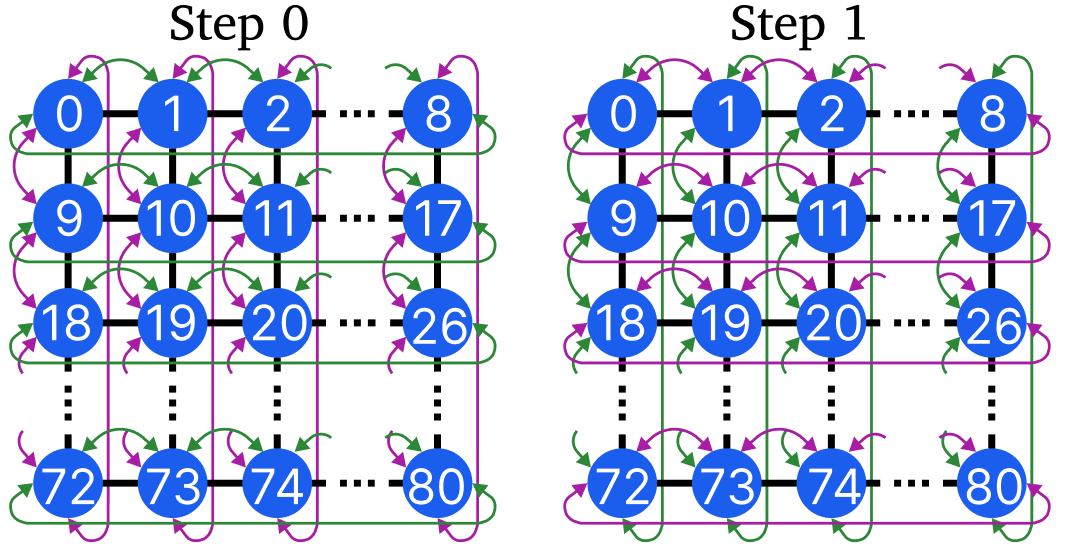}
    \Description{Two panels show the first two Trivance steps on a
    nine-by-nine torus. Green and purple arrows represent two data
    partitions communicating concurrently along different torus
    dimensions. In step zero one partition uses horizontal links while
    the other uses vertical links. In step one the partitions exchange
    dimensions, keeping their communication paths disjoint. Dotted
    lines indicate omitted intermediate nodes.}
  \caption{First two steps of \textsc{Trivance} AllReduce for a square torus of size $n=81$. Green arrows show the original collective; purple the mirrored collective.}
  \vspace{-10px}
  \label{fig:trivanceexamplemultid}
\end{figure}

Running AllReduce on multidimensional torus networks has no effect on the per step latency of all algorithms. The total transmitted data, and thus the bandwidth-optimality, remains also unchanged. In contrast, multidimensional torus networks of size~$n$ have a significant impact on the transmission delay optimality compared to rings of the same size. Table~\ref{tab:multi_optimality} presents the transmission delay optimality for \textsc{Trivance} and state-of-the-art algorithms. Since the closed-form expressions are hardly readable for the bandwidth-optimal algorithms, we also provide the asymptotic optimality values for $D=2,3,4$ as $n \to \infty$.

While the Bucket algorithm achieves optimal performance, the other bandwidth-optimal algorithms converge to an optimality factor of~1 as the dimension size increases. As mentioned prior, this emerges from reducing per-message size due to increased maximum injection bandwidth and avoiding congestion through shorter communication distances. For example, with $D=2$, Swing incurs 20\% higher cost than the optimum, and \textsc{Trivance} 33\%. As the dimensionality increases, this limitation diminishes for \textsc{Trivance}, dropping to 2\% for $D=4$, almost reaching the optimal transmission delay; Recursive Doubling (B) differs by 7\%. This shows that \textsc{Trivance} theoretically can match bandwidth-optimized approaches like Swing and Bucket in transmission delay optimality, while preserving its 50\% latency improvement.

\begin{table}[H]
\renewcommand{\arraystretch}{1.4} 
\centering
\begin{tabular}{l|c|c|c|c}
\textbf{Algorithm} & \textbf{Closed-form} & \textbf{D=2} & \textbf{D=3} & \textbf{D=4}\\
\hline
Rec.Doub. (L) & $D^2\sqrt[D]{n}$ & \cellcolor{bad}$4\sqrt{n}$ & \cellcolor{bad}$9n^{1/3}$ & \cellcolor{bad}$16n^{1/4}$  \\
Swing (L) & $\frac{D^2}{3}\sqrt[D]{n}$ & \cellcolor{bad}$\frac{4}{3}\,\sqrt{n}$ & \cellcolor{bad}$3\,n^{1/3}$ & \cellcolor{bad}$\frac{16}{3}\,n^{1/4}$ \\
Bruck (L) & $\frac{3D}{2}\cdot\sqrt[D]{n}$ & \cellcolor{bad}$3\sqrt{n}$ & \cellcolor{bad}$\frac{9}{2}n^{1/3}$ & \cellcolor{bad}$6n^{1/4}$ \\
\textsc{Trivance} (L) & $\frac{D}{2}\cdot\sqrt[D]{n}$ & \cellcolor{bad}$\sqrt{n}$ & \cellcolor{poor}$\frac{3}{2}n^{1/3}$ & \cellcolor{mid}$2n^{1/4}$ \\
\hline
Bucket & 1 & \cellcolor{opt}1 & \cellcolor{opt}1 & \cellcolor{opt}1 \\
Swing (B) & $\frac{2^{D}(2^{D}-1)}{(2^{D}-2)(2^{D}+1)}$ & \cellcolor{near}1.2 & \cellcolor{opt}1.04 & \cellcolor{opt}1.01  \\
\textsc{Trivance} (B) & $\frac{3^{D}-1}{3^{D}-3}$ & \cellcolor{near}1.33 & \cellcolor{opt}$1.08$ & \cellcolor{opt}$1.02$ \\
Rec.Doub. (B) & $\frac{2^D - 1}{2^D - 2}$ & \cellcolor{near}1.5 & \cellcolor{near}1.17 & \cellcolor{opt}1.07\\
Bruck (B) & $3\cdot\frac{3^{D}-1}{3^{D}-3}$ & \cellcolor{mid}4 & \cellcolor{mid}3.25 & \cellcolor{mid}3.06 \\
\end{tabular}
\caption{Transmission delay optimalities for $D \geq 2$ tori. The closed-form expressions are approximated for $n \to \infty$ with respect to the ideal cost $\frac{m \cdot \beta}{D}$.
}

\label{tab:multi_optimality}
\end{table}

\section{Evaluation}
\label{sec:evaluation}
In order to complement our analytical results derived in Section~\ref{sec:optimality} and Section~\ref{sec:multitori} (in particular latency-optimality), we assess in this section the performance of \textsc{Trivance} and state-of-the-art AllReduce algorithms empirically on different torus topologies. We conducted extensive simulations using the Structural Simulation Toolkit~\cite{janssen2010simulator} (SST), an event-driven, packet-level network simulator. Within SST, we implemented \textsc{Trivance} as well as Bruck's AllReduce algorithm. For the Bucket, Recursive Doubling, and Swing algorithms, we rely on the publicly available implementations by De Sensi et al.~\cite{295653}. In our evaluation, we extend Recursive Doubling to fully exploit all $2D$ network ports in multidimensional torus topologies. To ensure a reasonable comparison with the original Bruck algorithm, we modify its routing strategy to employ shortest-path routing, apply joint reductions introduced in this work and reorder data exchanges in the AllGather phase to reduce overall end-to-end transmission latency~\cite{jeaugey2025pat}; we detail these modifications in Appendix~\ref{bruckimplementation}. Baselines with a logarithmic step count provide a latency-optimal and a bandwidth-optimal variant; hence the comparison effectively involves both variants of \textsc{Trivance} and the corresponding variants of the baseline. For power-of-three network sizes, we compare \textsc{Trivance} to Bucket, Swing and Bruck. We omit Recursive Doubling as Swing consistently outperforms Recursive Doubling.

We simulate networks with link bandwidth of 800\,Gb/s, link latency of 100\,ns, packet size of 8192\,B, per hop packet processing latency of 100\,ns~\cite{295653} and introduce a per step latency of $\alpha=1.5\mu\,\mathrm{s}$~\cite{285084, zhao2025EffDCTopCC}. We vary these parameters in the sensitivity analysis presented in Appendix~\ref{app:sensitivity}. We evaluate AllReduce sizes ranging from 32\,B to 128\,MiB. Each plot reports the relative performance of existing approaches compared to \textsc{Trivance}. Traditional AllReduce algorithms are optimal for power-of-two network sizes, whereas \textsc{Trivance} is optimal for power-of-three sizes and remains resilient to suboptimal sizes. However, deviations from these introduce additional data transmission due to routing adjustments, which becomes significant for larger AllReduce sizes.

We organize the evaluation as follows: Section~\ref{eva:1D} presents results for ring networks, Section~\ref{eva:2D} for 2D tori with varying bandwidth and shape, and Section~\ref{eva:3D} for 3D tori. Appendix~\ref{appendix:add} provides additional results for different topology sizes. Our results show that \textsc{Trivance} is the best performing latency-optimal algorithm for AllReduce in latency-bound scenarios by 5-30\% with small and moderate AllReduce sizes. This advantage persists for AllReduce sizes up to 512\,KiB in ring networks, to 8\,MiB in multidimensional tori and up to 32\,MiB in high bandwidth networks. For 3D tori \textsc{Trivance} reduces completion time by up to 20\% compared to all state-of-the-art approaches up to 128\,MiB.
\newpage
\subsection{Performance on Rings}
\label{eva:1D}
For ring topologies, Figure~\ref{fig:torus8} and Figure~\ref{fig:torus64} show the performance of existing algorithms $A$ relative to \textsc{Trivance}, where completion time is normalized to that of \textsc{Trivance} ($\frac{(T_A-T_{\textsc{Trivance}})\cdot 100\%}{T_A})$ for AllReduce sizes ranging from 32\,B to 128\,MiB. Positive values denote an improvement of \textsc{Trivance} for a particular algorithm, while negative a degrade. We only report the best between the latency- and bandwidth-optimal versions. The results in Figure~\ref{fig:torus8} show that \textsc{Trivance} slightly reduces completion time compared to Bruck for small AllReduce sizes, while achieving more than a 20\% performance advantage over Swing and Recursive Doubling due to its improved per step latency. For larger AllReduce sizes, \textsc{Trivance} improves upon all existing approaches by up to 15\% at 128\,KiB. The tradeoff point at which Swing matches the performance of \textsc{Trivance} occurs at 512\,KiB; beyond this point, Swing outperforms \textsc{Trivance}. Starting at 4\,MiB, the Bucket algorithm achieves the lowest completion time.

\begin{figure}[t]
    \centering
    \includegraphics[width=\columnwidth]{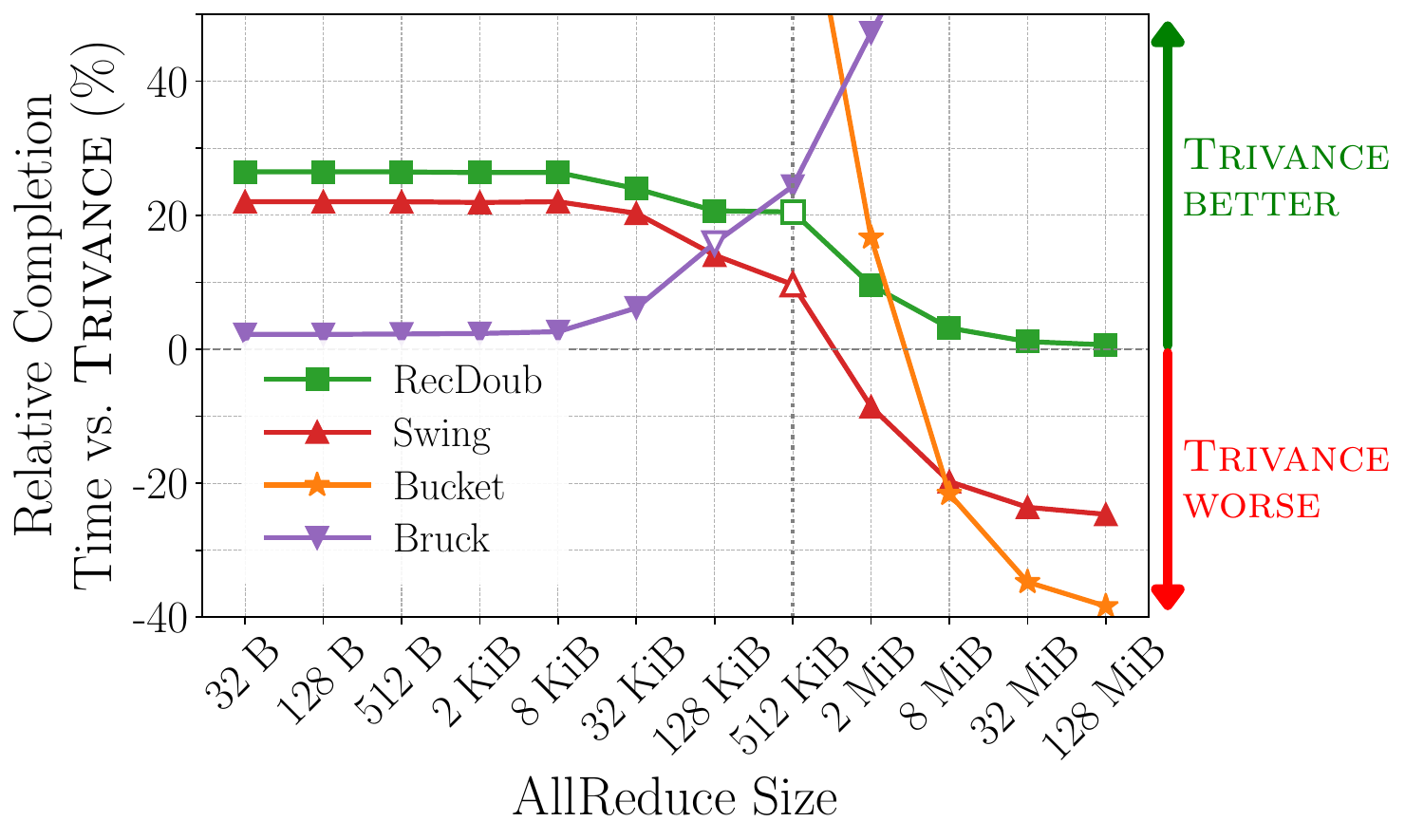}
    \setlength{\abovecaptionskip}{-5pt}
    \Description{Line chart of relative AllReduce completion time on an
    eight-node ring for message sizes from 32 B to 128 MiB. Positive
    values mean that Trivance is faster. For small messages, Trivance is
    more than 20 percent faster than Recursive Doubling and Swing and
    slightly faster than Bruck. Swing reaches the tradeoff near 512 KiB, and
    bandwidth-optimized algorithms become faster for large
    messages. Hollow points mark changes between latency- and
    bandwidth-optimal variants.}
    \caption{AllReduce completion time of state-of-the-art algorithms relative to \textsc{Trivance} for AllReduce sizes of 32\,B to 128\,MiB on a ring of size 8. Hollow data points denote the transition from latency- to bandwidth-optimal versions and the dotted vertical does this for \textsc{Trivance}.}
    \label{fig:torus8}
    \vspace{-5px}
\end{figure}
\begin{figure}[t]
        \centering
        \includegraphics[width=\columnwidth]{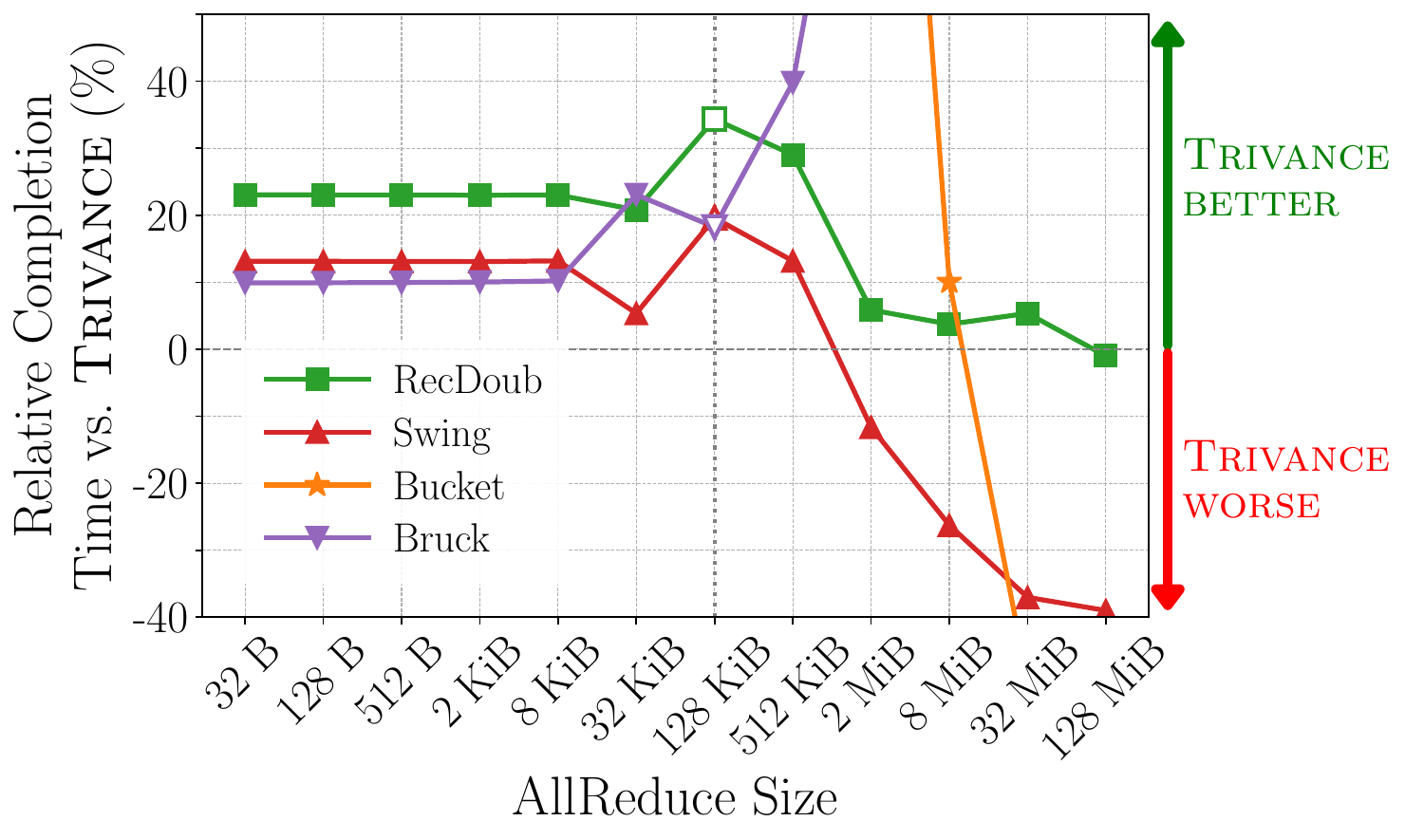}
        \setlength{\abovecaptionskip}{-5pt}
        \Description{Line chart of relative AllReduce completion time on a
    64-node ring. Trivance is faster than the evaluated algorithms for
    small messages, while performance fluctuates around the transition
    between latency- and bandwidth-optimal variants. Swing becomes
    faster for larger messages, whereas Bruck is substantially slower
    than Trivance in the intermediate-size range.}
        \caption{AllReduce completion time for a 64 node ring.}
        \label{fig:torus64}
        \vspace{-10px}
\end{figure}

\subsection{Performance on 2D Tori}
As discussed in Section~\ref{sec:multitori}, the transmission delay optimality of bandwidth-optimal algorithms converges for high dimensional torus networks toward the theoretical optimum. In contrast, latency-optimality remains unaffected, as shown in Figure~\ref{fig:torus32x32}, where the latency improvement of \textsc{Trivance} for small messages is similar to Figure~\ref{fig:torus8}. Due to improved transmission delay optimality, the tradeoff point at which \textsc{Trivance} is outperformed by existing approaches shifts to higher AllReduce sizes, reaching 16\,MiB in Figure~\ref{fig:torus32x32}. Moreover, in $32\times32$ torus networks, \textsc{Trivance} outperforms all existing approaches by up to 20\% for AllReduce sizes between 32\,KiB and 2\,MiB. Within this span, AllReduce sizes are too large for Bruck to remain effective, yet too small for bandwidth-optimized algorithms such as Swing or Bucket to achieve minimal completion time. 

Notably, the performance of \textsc{Trivance} and Bruck on power-of-two networks strongly depends on how closely the network size matches their optimal topologies. When the network size deviates from these, both algorithms transmit data volumes comparable to those required for the next larger power-of-three topology, resulting in relatively higher communication overhead. While this overhead is negligible for small AllReduce sizes, it becomes significant for larger messages, particularly beyond 8\,MiB.

\noindent\textbf{Rectangular Tori:} The performance of \textsc{Trivance} in Figure~\ref{fig:rectangular} on rectangular tori follows the same trend as on square tori: \textsc{Trivance} performs best in the latency-dominated regime with the tradeoff point occurring around 8\,MiB, after which Swing becomes preferable. Notably, Swing remains the strongest baseline, since Bucket performs substantially worse on rectangular tori~\cite{295653}.

\noindent\textbf{Bandwidth Impact: }
\label{eva:BW}
To analyze the impact of network bandwidth on the performance of \textsc{Trivance}, Figure~\ref{fig:bandwidth} illustrates the relative completion time of the best-performing existing approach for each AllReduce size compared to \textsc{Trivance} on a $32\times32$ torus. The evaluated network bandwidths range from 200\,Gb/s to 3.2\,Tb/s. We observe that the performance improvement of \textsc{Trivance} for small AllReduce sizes is largely independent of bandwidth, as this regime is latency bound. In this setting, \textsc{Trivance} reduces completion time by 14\% for AllReduce sizes up to 2\,MiB.

At lower bandwidths, transmission delay constitutes a larger fraction of the overall completion time, causing bandwidth optimized approaches to surpass \textsc{Trivance} at approximately 4\,MiB. As bandwidth increases, the relative impact of congestion diminishes, allowing the performance advantage of \textsc{Trivance} to persist for larger AllReduce sizes, extending up to 64\,MiB for networks operating at 2.4\,Tb/s and 3.2\,Tb/s. These results indicate that \textsc{Trivance} consistently achieves better performance at high bandwidth, whereas bandwidth-optimized algorithms dominate in low-bandwidth regimes. Regardless of network bandwidth, \textsc{Trivance} remains the best-performing latency-optimal algorithm.

\label{eva:2D}
\begin{figure}[t]
\centering
\includegraphics[width=\columnwidth]{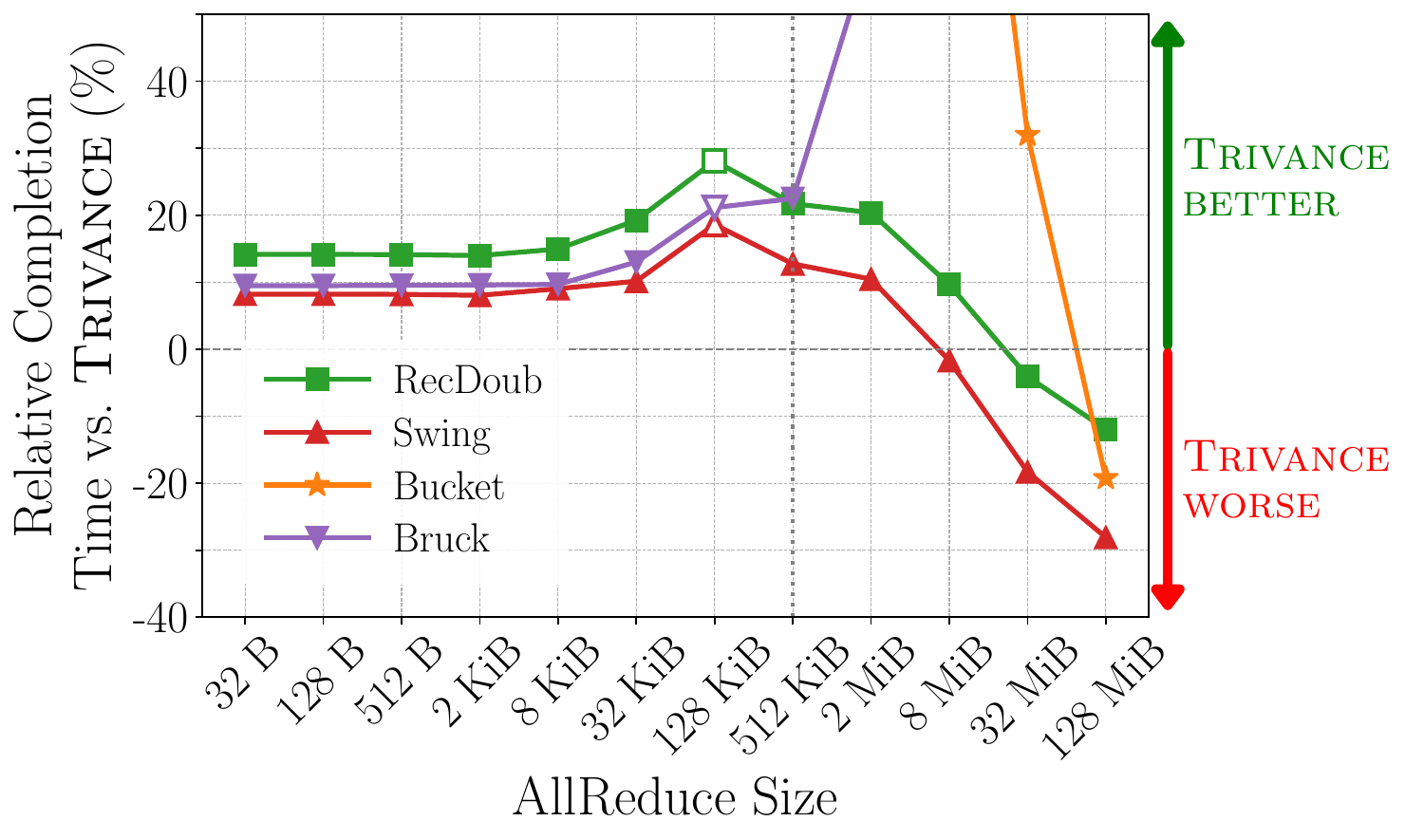}
\Description{Line chart of relative AllReduce completion time on a
32x32 torus. Trivance is approximately 10 to 30 percent faster
than the baselines for small and intermediate messages, with its
largest gains around tens to hundreds of KiB. Its advantage remains
positive into the large AllReduce sizes, after which Recursive Doubling,
Swing, and eventually Bucket become competitive.}
\caption{AllReduce completion time for $32\times32$ tori.}
\label{fig:torus32x32}

\end{figure}

\begin{figure}[t]
    \centering
    \includegraphics[width=\columnwidth]{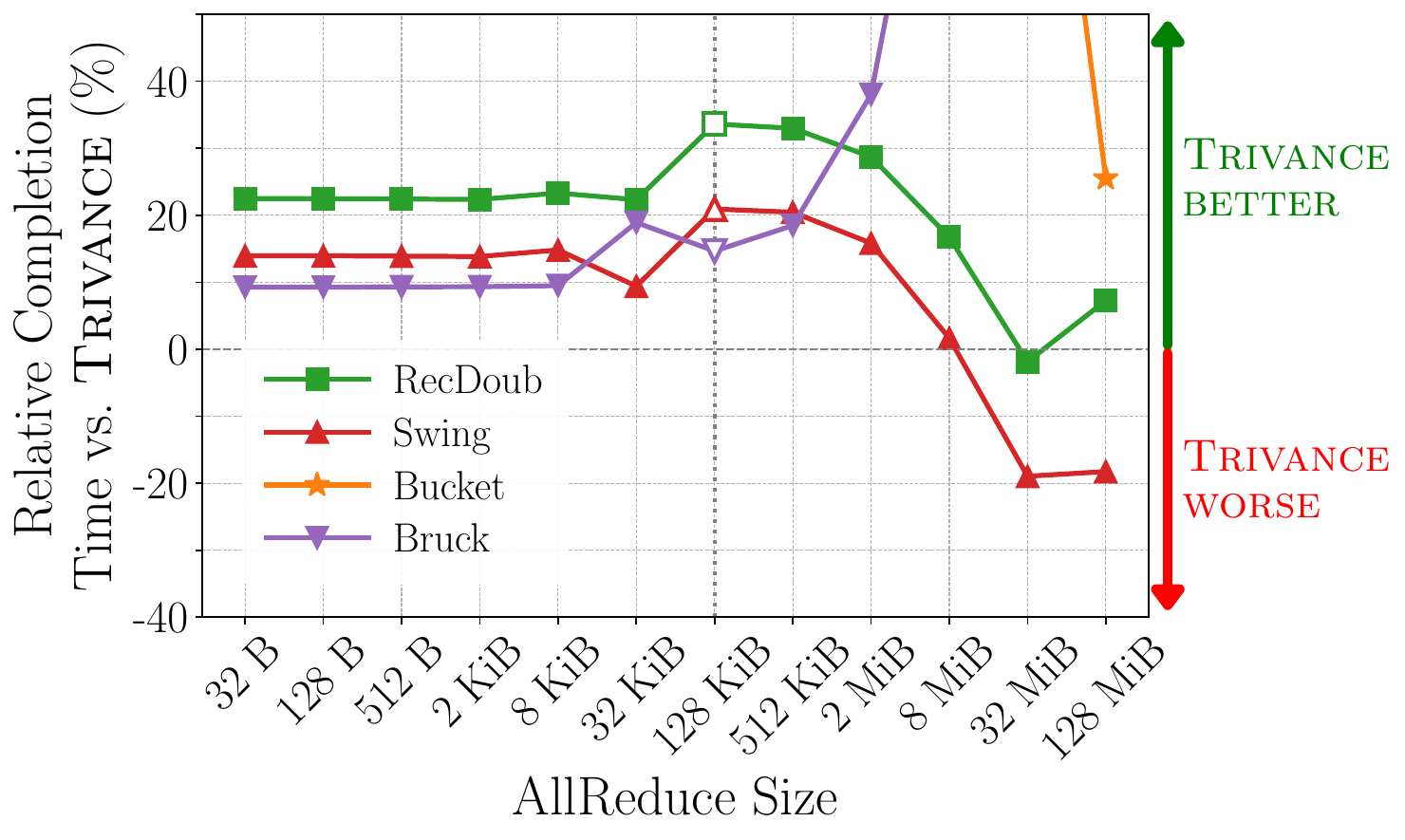}
    \Description{Line chart of relative AllReduce completion time on a
    rectangular 8x64 torus. Trivance is roughly 10 to 35 percent
    faster than the logarithmic baselines for small and intermediate
    messages. Swing reaches the trade-off point at approximately 8 MiB and becomes
    faster for larger messages. Recursive Doubling also becomes faster
    at large sizes, while Bucket performs poorly at the largest plotted
    size.}
  \caption{AllReduce completion time relative to \textsc{Trivance} on a rectangular $8\times64$ torus.}
  \label{fig:rectangular}
\end{figure}

\noindent\textbf{Power-of-Three Tori: }
\label{eva:3k}
The communication pattern of \textsc{Trivance} supports AllReduce on arbitrary network sizes. However, its optimal performance is achieved on networks with dimension sizes of powers of three, under which \textsc{Trivance} operates most efficiently. In Figure~\ref{fig:27x27}, we analyze the completion time improvement of \textsc{Trivance} on a $27\times27$ torus network. For small AllReduce sizes, \textsc{Trivance} consistently performs better than Bruck or Swing. Notably, for larger AllReduce sizes, \textsc{Trivance} significantly outperforms all baselines, achieving a 10\% performance gain for messages larger than 512\,KiB and exceeding 50\% for messages larger than 2\,MiB. Even at 32\,MiB, \textsc{Trivance} continues to outperform all baselines by more than 40\%, with Bucket only matching the performance of \textsc{Trivance} at 128\,MiB. These results highlight the advantage of \textsc{Trivance} over all existing approaches on power-of-three network topologies for messages up to 128\,MiB.

\newpage
\begin{figure}[H]
    \centering
    \includegraphics[width=\columnwidth]{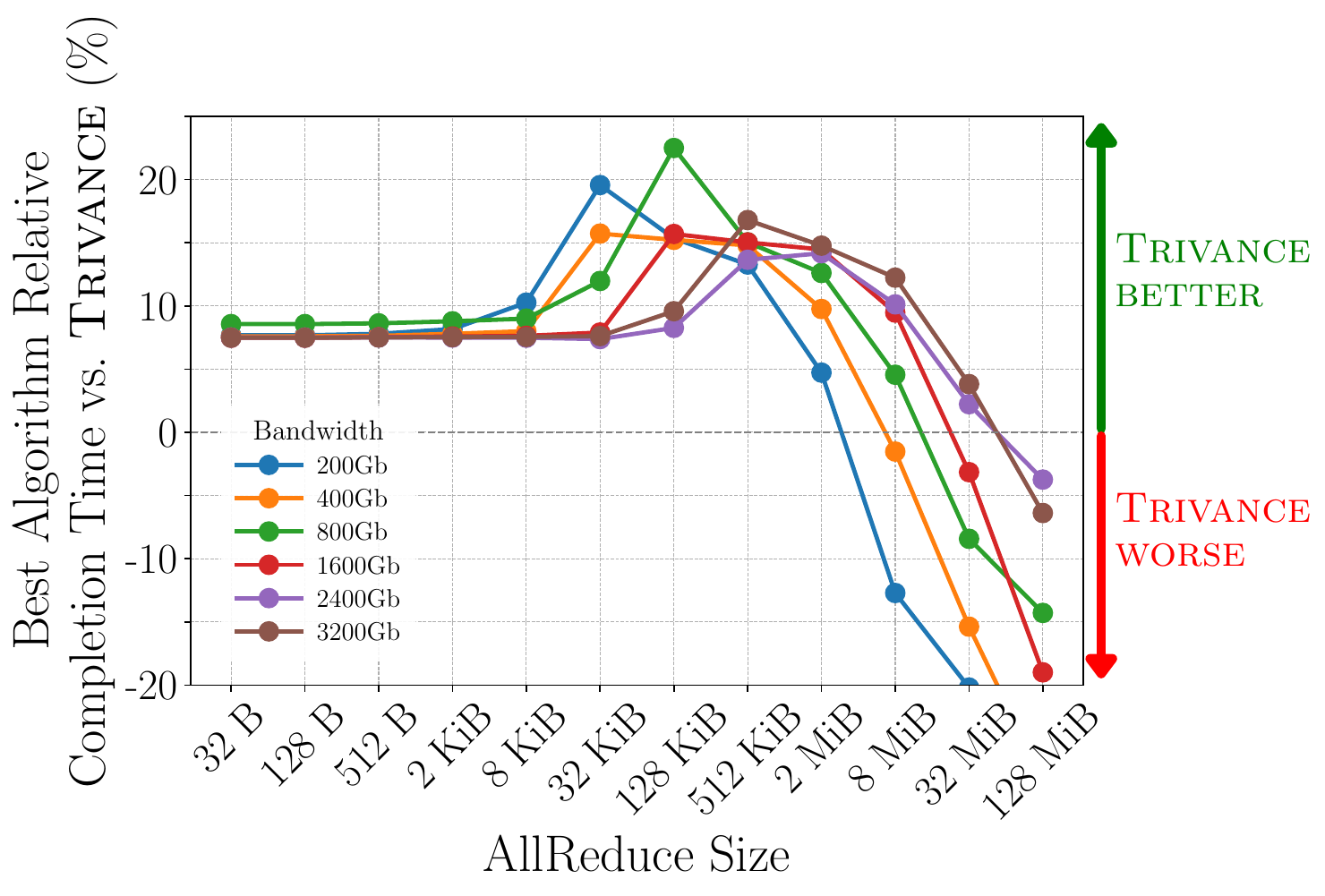}
    \Description{Six curves compare Trivance with the best baseline on a
    32x32 torus for link bandwidths from 200 Gb/s to 3.2 Tb/s.
    Performance is similar across bandwidths for small messages.
    Increasing bandwidth moves the peak advantage and the crossover
    with the best baseline toward larger messages. Low-bandwidth
    networks favor the baseline at relatively small MiB sizes, whereas
    the highest bandwidths preserve the Trivance advantage for large AllReduce sizes.}
  \caption{AllReduce completion time on a $32\times32$ torus with varying network bandwidth from $200$\,Gb/s to $3.2$\,Tb/s. Each graph compares, for a specific bandwidth, \textsc{Trivance} to the best performing algorithm.}
  \label{fig:bandwidth}
\end{figure}

\begin{figure}[H]
    \centering
    \includegraphics[width=\columnwidth]{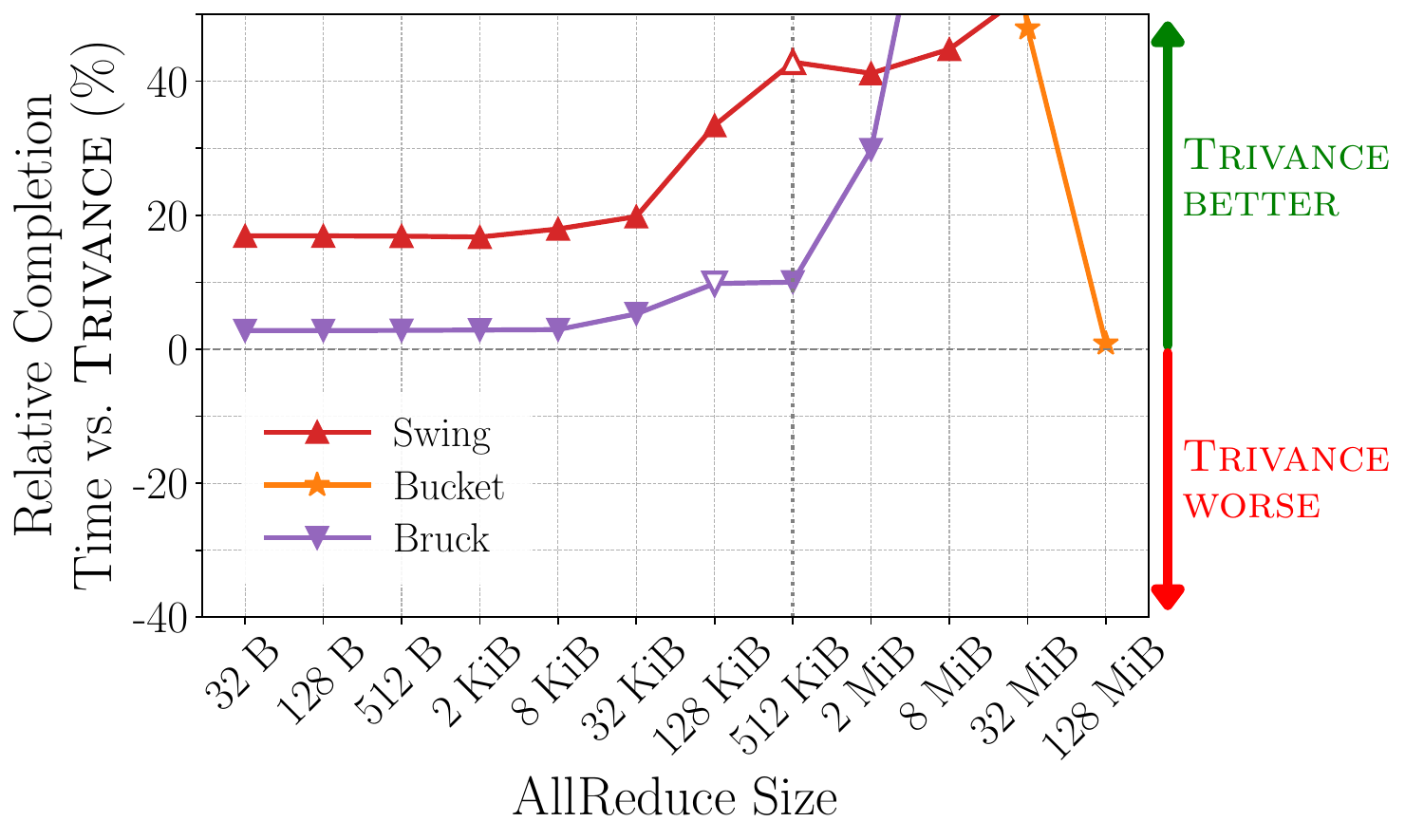}
    \Description{Line chart comparing Swing, Bucket, and Bruck with
    Trivance on a 27x27 torus. Trivance is faster throughout nearly
    the entire range. Its advantage over Swing increases from about
    15 percent for small messages to more than 40 percent for large AllReduce sizes, while Bruck becomes increasingly slower.
    Bucket approaches the trade-off point only at 128 MiB.}
  \caption{AllReduce completion on a $27\times27$ torus of Bucket, Swing and Bruck compared to \textsc{Trivance}.}
  \label{fig:27x27}
\end{figure}

\begin{figure}[H]
    \centering
        \includegraphics[width=\columnwidth]{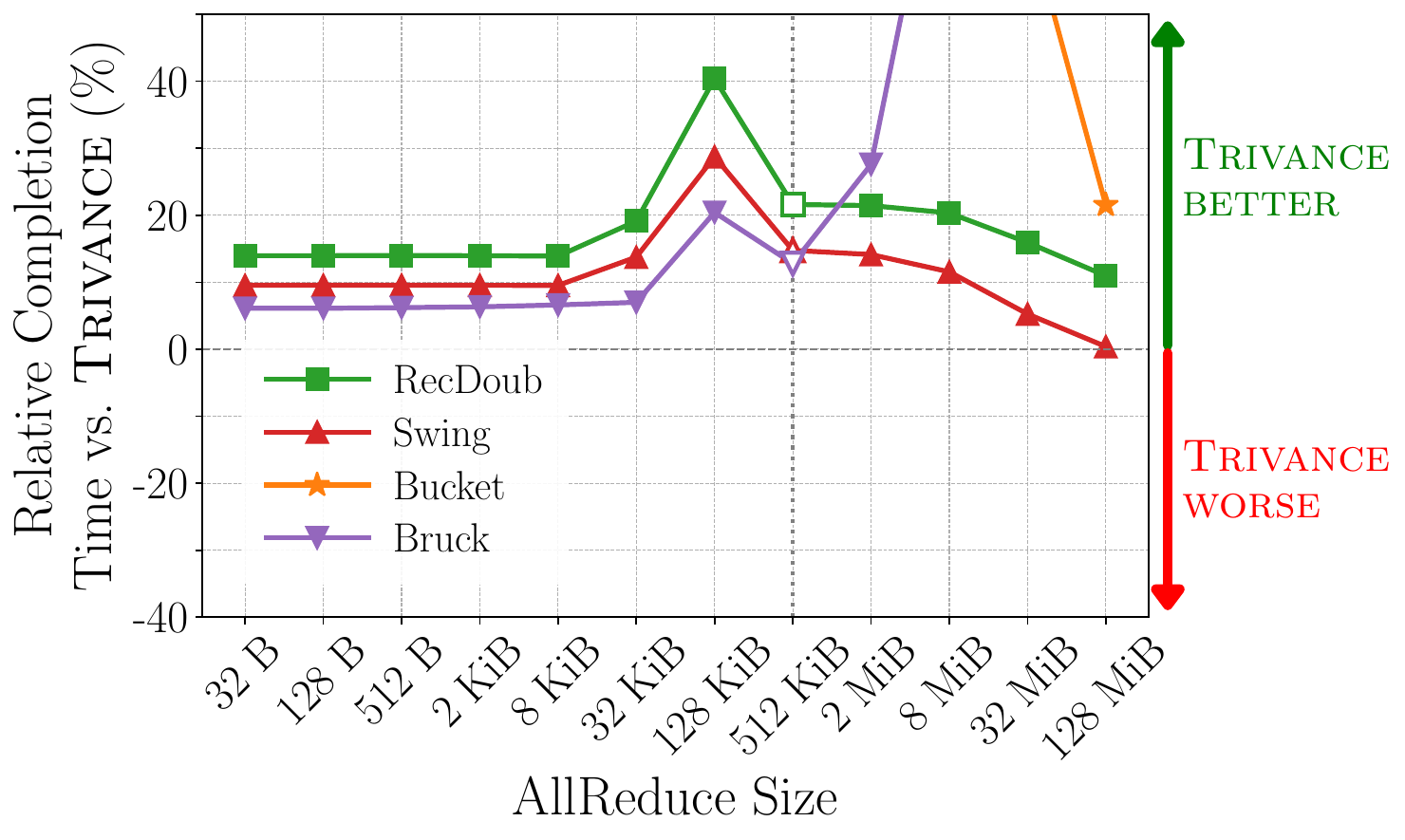}
    \Description{Line chart of relative AllReduce completion time on a
    16x16x16 torus. All baseline curves remain above zero,
    indicating that Trivance is faster across the complete range from
    32 B to 128 MiB. The largest improvements occur around 128 KiB,
    and Trivance retains an advantage of approximately 5 to 20 percent
    for the largest messages.}
  \caption{AllReduce completion on a $16\times16\times16$ torus.}
  \label{fig:16x16x16}
\end{figure}
\newpage
\subsection{Performance on 3D Tori}

\label{eva:3D}
Figure~\ref{fig:16x16x16} shows that, consistent with the theoretical bounds in Section~\ref{sec:multitori}, the transmission delays of all bandwidth-optimal algorithms approach the optimum. Thus, improvements in link and per step latency become the dominant factor. On a $16\times16\times16$ torus, \textsc{Trivance} outperforms all existing approaches across the entire AllReduce size spectrum up to 128\,MiB, achieving an up to $20\%$ reduction in completion time. In multidimensional networks, most data is transmitted under near-optimal conditions, so algorithms that minimize per step latency and shortcut the ring by connecting peers at minimal distance perform best. \textsc{Trivance} exploits both directions of each dimension to minimize communication distance and jointly reduces data to complete in $\log_3 n$ steps. Even at 128\,MiB, it reduces the total completion time by $8\%$ compared to the second best algorithm, Swing. We expect this trend to persist for 4D torus networks.

\subsection{Summary}

\textbf{Latency- and Bandwidth-Bound Regimes:} \textsc{Trivance} reduces both the number of communication steps and the aggregate path length, yielding its largest gains in the latency-bound regime, where latency-optimal versions are selected. Between $\approx$ 128\,KiB and 512\,KiB, the algorithms switch to their bandwidth-optimal variants as transmission time becomes increasingly dominant. Nevertheless, per step and per hop latency remain non-negligible, so \textsc{Trivance} continues to benefit for bandwidth-optimal versions. Figure~\ref{fig:rectangular} shows that the latency-optimal \textsc{Trivance} variant begins to degrade at 32\,KiB, but switching to the bandwidth-optimal variant preserves its advantage up to 8\,MiB. This is possible because the bandwidth-optimal variant preserves \textsc{Trivance}'s step count advantage while reducing transmitted data volume compared to its latency-optimal version, making it effective for intermediate AllReduce sizes.

\textbf{When is \textsc{Trivance} beneficial?} \textsc{Trivance} delivers consistently the best performance by $5-30\%$ in latency-bound regimes for AllReduce sizes up to 1\,MiB as it improves over existing AllReduce approaches by per step latency and path length. For intermediate AllReduce sizes from 1\,MiB to 8\,MiB, the benefit depends on the network parameters. Higher link bandwidth, per hop and per step latency shift the tradeoff point to larger AllReduce sizes up to which \textsc{Trivance} remains the best performing algorithm by up to $10\%$. This advantage is reinforced as network dimensionality increases (Figure~\ref{fig:torus32x32}). For workloads $\geq$16\,MiB, \textsc{Trivance} is only beneficial with sufficiently high bandwidth or in 3D tori. In Figure~\ref{fig:16x16x16}, \textsc{Trivance} remains beneficial across the full AllReduce size range, improving over mirrored approaches by up to $30\%$ and over adapted Bruck by up to $20\%$. For large AllReduce sizes in rings or 2D tori, Bucket remains best. In conclusion, \textsc{Trivance} is the best performing algorithm for small and intermediate AllReduce sizes and remains advantageous as network dimensionality increases. This is particularly relevant for systems such as TPUv4 pods, whose ICI fabric uses a three-dimensional torus topology.

\section{Discussion}
\label{sec:discussion}

\paragraph{Implementation Feasibility and Limitations of Network Simulation}
Our packet-level SST evaluation isolates the effects of topology, routing, congestion, message size, latency, and bandwidth on \textsc{Trivance}'s communication schedule, but does not replace an implementation in a production collective library. Such an implementation is feasible because the schedule is static and can be expressed through asynchronous sends, receives, and local reductions. The main implementation challenge is coordinating the two incoming streams: matching chunks must be identified, their readiness tracked, and the next step started only after both inputs have been combined with the local partial result. The implementation must also map the two dependent communication directions to the correct distinct NICs or channels so that both streams can progress concurrently without serialization through a shared injection path.

A complete implementation for a large physical torus system remains future work and is a limitation of our evaluation. Such an implementation could expose additional effects beyond our simulator, including channel and protocol selection, NIC mapping, chunk-level buffering/synchronization and readiness tracking for both incoming receives. We further discuss the implementation path in Appendix~\ref{integrationNCCL}.

\paragraph{Joint Reduction Overheads}
\textsc{Trivance} requires a stronger per step synchronization upon receiving transmission than mirrored independent collectives---each node must receive both streams from each direction to proceed. Since regular ring and torus topologies are considered left--right symmetric, both peers from which data is received are located at same distance with the same congestion. Together with our homogeneous network capacity and computation assumption, all concurrent transmissions have the same expected completion time. Given an individual transmission is delayed, the corresponding paired transmission in \textsc{Trivance} needs to wait until proceeding. However, for mirrored collectives with independent progression, the completion time also remains bounded by the slowest sub-collective, which is mapped to distinct ports to avoid congestion~\cite{295653}. A delay in one sub-collective can desynchronize their schedules and cause previously separated traffic to overlap, introducing additional congestion.

For further optimization by chunk level pipelining, joint reductions require additional coordination, buffering, and readiness tracking. In \textsc{Trivance}, a chunk received from one direction can be reduced only after the matching chunk from the other direction has also arrived. The pipeline therefore requires a readiness condition for each chunk pair, as further discussed in Section~\ref{integrationNCCL}. Consider a message divided into chunks $1,2,\ldots,z$. In \textsc{Trivance}, only chunks with the same index can participate in a joint reduction. If the left peer transmits chunks in ascending order while the right in descending, approximately half of each message may accumulate before the first matching pair becomes available. With aligned ordering, matching chunks become available in the same sequence, avoiding systematic accumulation of unmatched data.

\paragraph{Reduction Volume}\textsc{Trivance} reduces the total reduced byte-volume from approximately $2m$ to $1.5m$, a $25\%$ reduction. Both incoming directions are combined with the same active partial state, whereas mirrored binary collectives reduce two disjoint sub-collectives independently. For bandwidth-optimal \textsc{Trivance}, each node receives two messages of size $\frac{m}{3^{k+1}}$ per step and reduces them with the local partial result of size $\frac{m}{3^k}$ which sums asymptotically to $\sum_{k=0}^{\log_3 n-1} \frac{m}{3^k}=1.5m(1-\frac{1}{n})$ per Reduce-Scatter phase. Ring and Bucket reduce $\frac{m}{2n}$ bytes, resulting in total reduced byte-volume $2m(1-\frac{1}{n})$; Swing and Recursive Doubling, each mirrored sub-collective reduces $\frac{m}{2\cdot 2^{k+1}}$ bytes in step $k$, which gives a reduced byte-volume $\frac{2m}{2^{k+1}}$ per step, totaling to $2m(1-\frac{1}{n})$.

\paragraph{Applicability to Other Collectives and Topologies}
While \textsc{Trivance} is designed primarily as an AllReduce algorithm for tori, the underlying communication pattern is not limited to these assumptions. Collective primitives that rely on existing algorithms from Section~\ref{sec:algorithms} operating on multiport topologies can adopt \textsc{Trivance}'s routing and data-partitioning paradigm for latency optimization. \textsc{Trivance} could also be mapped to Clos and hierarchical fabrics in which each node has multiple independent network paths. However, deriving the corresponding theoretical guarantees, routing strategy, and performance characteristics remains future work.

\paragraph{Latency-Optimality in Multidimensional Tori}
In multidimensional torus networks, it is theoretically possible to simultaneously utilize all available ports to complete in $\log_{2D+1} n$ steps. However, implementing such a pattern requires highly specific dimension sizes for optimal performance and torus deployments typically reserve different dimensions for distinct communication patterns or workloads.

\section{Related Work}
\label{sec:relatedwork}
This work builds on prior research on AllReduce algorithms for torus networks~\cite{295653,10.1145/2686882,10.1007/978-3-540-24685-5_1, jeaugey2025pat}. It introduces a novel latency-optimal algorithm as a bidirectional extension of Rabenseifner's classic Recursive Doubling approach~\cite{10.1007/978-3-540-24685-5_1}. Sack and Gropp~\cite{10.1145/2686882} were among the first to analyze collective communication algorithms on torus networks, proposing a novel parallel multidimensional communication pattern for both Bucket and Recursive Doubling. Träff and Hunold~\cite{traff2020MultiMPI} similarly advocate exploiting parallelism across independent communication resources by decomposing MPI collectives into multiple concurrent sub-collectives.

In another direction, scheduling based approaches~\cite{10.1145/3718958.3750499, 10.1145/3718958.3750514, 10.1145/3651890.3672249, 285084, 10.1145/3437801.3441620} address collective communication bottlenecks by casting collective operations as a routing and scheduling optimization problem that is explicitly parameterized by the underlying topology, communication costs, and resource constraints. While effective at improving bandwidth utilization and adapting execution to heterogeneous environments, these approaches generally operate over fixed collective algorithms or focus on schedule synthesis rather than redesigning the algorithmic communication pattern itself. 

A large number of supercomputers are built from torus networks~\cite{TorusSupercomputer, 10.1145/2686882} and they are widely used running machine learning workloads~\cite{zu2024resiliency, 10.1145/3623490}. Despite offering lower bisection and global bandwidth than topologies like Clos, torus networks provide a cost-effective solution for workloads such as ML training, where communication is often structured as a 3D logical torus~\cite{vahdat2023enabling, TPUv4, fu2024distributed}. For example, Google’s TPUv4 employs up to 4,096 chips in a 3D (twisted) torus topology supporting dynamic reconfiguration by optical circuit switches~\cite{TPUv4}.
\section{Conclusion}
We presented \textsc{Trivance}, a novel AllReduce algorithm for direct-connect topologies with bidirectional links that completes in $\log_3 n$ steps while reducing congestion by a factor of three compared to existing latency-optimal approaches. By leveraging the additional communication port available at each node, \textsc{Trivance} increases the communication distance per step by a factor of three, thereby shortcutting the ring and accelerating AllReduce completion relative to traditional algorithms like Swing or Recursive Doubling. We formally proved the correctness of \textsc{Trivance} and derived latency, bandwidth, and transmission delay lower bounds for state-of-the-art AllReduce algorithms. We further evaluated \textsc{Trivance} across a wide range of ring and multidimensional torus topologies, demonstrating consistent performance improvements in latency-bound regimes and for 3D tori. Depending on dimensionality, network size, and network parameters, these improvements persist up to 32\,MiB. These properties make \textsc{Trivance} particularly well suited for emerging high-dimensional torus-based GPU and accelerator clusters, such as TPUv4.

Our work opens several avenues for future research. In particular, we believe that the \textsc{Trivance} approach may also improve performance on other multiport network topologies, including reconfigurable networks, where reducing the number of communication steps can lower reconfiguration overhead~\cite{juerss2026bruck}. Exploring this direction for AllReduce and other collectives remains future work.\\

\noindent \textbf{{\large Acknowledgments}}

\medskip
\noindent We thank Volker Stocker for valuable discussions and his feedback. We are also grateful to Jesper Larsson Träff and Daniele De Sensi for helpful inputs concerning the implementation. This work is part of a project that has received funding from the German Federal Ministry  of  Research,  Technology  and  Space  (BMFTR)  under grant 16DII141 “Weizenbaum Institut für die vernetzte Gesellschaft”, the European Research Council (ERC)---project FortifyNet (grant 101287293), 2026-2027 and NSF Future CoRe (grant 2551372).

\begin{figure}[!h]
    \centering
    \includegraphics[width=0.6\linewidth]{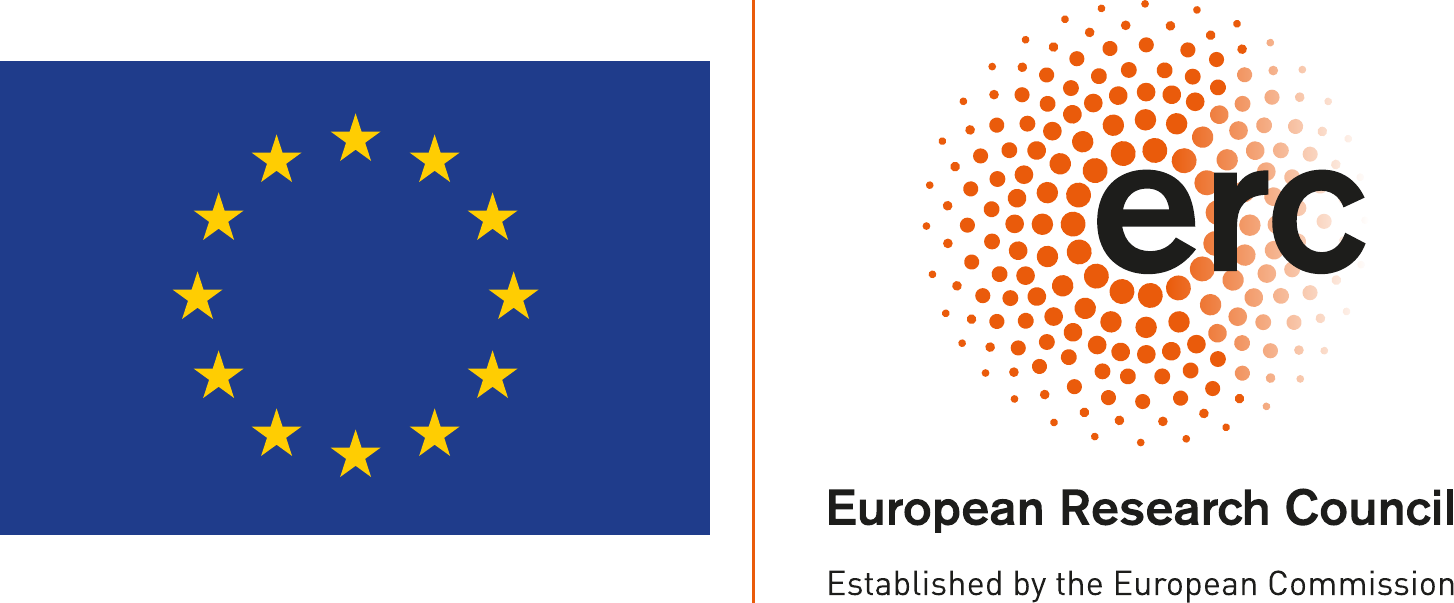}
    \Description{European Research Council and European Union logos.}
    \label{fig:my_label}
\end{figure}

\newpage
\bibliographystyle{ACM-Reference-Format}
\bibliography{reference}

\newpage 

\appendix

\noindent Appendices are supporting material that has not been peer-reviewed.

\section{Latency-Optimality Proof for \textsc{Trivance}}
\label{A:fullproof}
We have established Lemma~\ref{lem:range}, which characterizes the set of nodes $r$ has obtained data from after step~$k$. We now prove that this range covers the entire network after $\log_3 n$ steps.
\begin{proof}
We proceed by induction on $k$.

\noindent Let $D^k(r) \subseteq \{0,\dots, n-1\}$ represent the indices of nodes whose blocks node $r$ has received up to and including step $k$.

\vspace{1em}
\noindent\textbf{Base Case ($k = 0$)}:\\
Each node $r$ sends the respective blocks of data to nodes at distance $3^0 = 1$:
\[
D^0(r) = \{ (r-1) \bmod n,\, r,\, (r+1) \bmod n \}
\]

This covers nodes within distance of $R_0 = 1$, which matches the definition: $R_0 = \frac{3^1 - 1}{2} = 1$.

\noindent\textbf{Inductive Hypothesis:}\\
Assume that after step $k$, each node $r$ has received vectors from all nodes within distance $R_k = \sum_{i=0}^{k} 3^i = \frac{3^{k+1} - 1}{2}$.
That is:
\[
D^k(r) = \{ u \in \{0,\dots, n-1\} \mid \text{distance}(u, r) \le R_k \}
\]

\vspace{1em}
\noindent\textbf{Inductive Step:}\\
At step $k+1$, each node $r$ receives data from nodes at distance $3^{k+1}$ on both sides. The communication peers from $r$ at step $k+1$ can be defined as:
\begin{center}
    $peer_{\text{right}} = (r + 3^{k+1}) \bmod n, \quad$
    $peer_{\text{left}} = (r - 3^{k+1}) \bmod n$
\end{center}
The distance between $r$ and its peers is $3^{k+1}$, meaning there are $3^{k+1} - 1$ nodes between $r$ and each peer. According to the inductive hypothesis, node $r$ received data from all nodes within distance of $R_k = \sum_{i=0}^k 3^i$, and the same applies to its peers respectively. The set of nodes between $r$ and $\text{peer}_{\text{right}}$, $3^{k+1} - 1$, can be partitioned into two disjoint ranges of size $\sum_{i=0}^{k} 3^i$:
\[
3^{k+1} - 1 = 2 \cdot \sum_{i=0}^{k} 3^i
\]
Since $peer_{\text{right}}$ has also received all blocks within a distance of $\sum_{i=0}^{k} 3^i$, the blocks located between $r$ and $peer_{\text{right}}$ must be known to either $r$ or $peer_{\text{right}}$. Since the $\sum_{i=0}^{k} 3^i$ closest nodes to $r$ are already known by $r$, and the $\sum_{i=0}^{k} 3^i$ closest nodes to $peer_{\text{right}}$ are already known by $peer_{\text{right}}$, the intermediate nodes can be partitioned between them without overlap. No block can be simultaneously known by both or by neither. As a result, the sets of nodes from which $r$ and $peer_{\text{right}}$ have respectively received blocks are disjoint within the range between them. The same argument applies to the pairs $(r$,  $peer_{\text{left}})$.

As a result, during step $k+1$, node $r$ receives: all designated data held by $\text{peer}_{\text{right}}$, covering $\sum_{i=0}^{k} 3^i$ nodes to its right and left, plus $\text{peer}_{\text{right}}$'s own data. Likewise from $peer_{\text{left}}$, data from another $3^{k+1}$ nodes (totaling unseen data from $3^{k+1}$ nodes). Since the data from the two peers is disjoint and strictly new to $r$, its range $R_k = \sum_{i=0}^{k} 3^i$ at step $k$ extends by exactly $3^{k+1}$ nodes from its peers in both directions:
\[
R_{k+1} = R_k + 3^{k+1} = \sum_{i=0}^{k} 3^i + 3^{k+1} =\sum_{i=0}^{k+1} 3^i = \frac{3^{k+2} - 1}{2}
\]
This completes the inductive step and proves the claim.
\end{proof}
This proves that after each step $k$ according to \textsc{Trivance}'s communication pattern, each node holds blocks from all nodes within range $\sum_{i=0}^{k} 3^i$. Theorem~\ref{the:steps} establishes that after $\log_3 n$ steps this range covers the entire network of size $n$. The following proves that by contradiction. 
\begin{proof}[Proof by contradiction]
Assume for contradiction that after $s$ steps of the \textsc{Trivance} algorithm, there exists a node $r \in \{0,\dots, n-1\}$ and a node $u \in \{0,\dots, n-1\}$ such that node $r$ has not received data from $u$. That is $u \notin D^{s-1}(r)$. According to Lemma~\ref{lem:range}, node $r$ holds after step $s-1$ data from all nodes within distance $R_{s-1} = \frac{3^s - 1}{2}$. Since $n = 3^s$, the ring has a diameter of $\left\lfloor \frac{n}{2} \right\rfloor = \frac{3^s - 1}{2} = R_{s-1}$. This implies that any node in the ring is at most distance $R_{s-1}$ away from node $r$, due to the symmetry of the ring. Therefore, every node $u$ satisfies:
\[
\text{distance}(u, r) \le R_{s-1}
\]
This suggests that any node $u$ is at most distance $R_{s-1}$ from $r$, and thus must belong to $D^{s-1}(r)$. This contradicts the assumption that $u \notin D^{s-1}(r)$. Therefore, the assumption is false and it follows that by the end of step $s-1$, node $r$ has received the data blocks $b_r^u$ from all nodes $u \in \{0,\dots, n-1\}$ as:
\[
\forall r \in \{0,\dots, n-1\}, \quad D^{s-1}(r) = \{0,\dots, n-1\}
\]
\end{proof}
\section{Calculations of Transmission delay bounds}
Transmission delay optimality is obtained by normalizing the transmission delay with respect to the optimal term $m \cdot \beta$. This expression corresponds to the fraction of the data vector~$m$ that is transmitted, multiplied by the network congestion. 
\subsection{For Ring Networks} The Ring algorithm transmits the minimal amount of data across both available ports without message overlap, yielding an optimality factor of~1. For the Recursive Doubling (L), the congestion in step $k$ is $2^k$, which sums up for $\log_2 n$ steps to a transmission delay optimality of $\sum_{k=0}^{\log_2 n - 1} 2^k = 2^{\log_2 n} - 1 \approx n$, via geometric summation. In the bandwidth-optimal version, congestion is $2^k$ while data size halves as $1/2^{k+1}$, giving a constant product of $1/2$ per step and a total transmission delay optimality of $\tfrac{1}{2}\log_2 n$. Since the cost is multiplied by 2 for both phases, but the transmitted data is divided equally between the two ports, the factors cancel out. In the latency-optimal Swing algorithm, the distance is $\tfrac{2^{k+1}-(-1)^{k+1}}{3}$~\cite{295653}. With half the nodes transmitting clockwise and half counterclockwise, the congestion is $\left\lceil \tfrac{2^{k+1}-(-1)^{k+1}}{6} \right\rceil$, which converges for to $\tfrac{n}{3}$ as $n\to\infty$. In the bandwidth-optimal Swing, the mirrored collective doubles the congestion to $\tfrac{2^{k+1}-(-1)^{k+1}}{3}$, while the data size per step decreases to $1/2^{k+1}$, yielding a total that converges to $\tfrac{\log_2 n}{3}$. Regarding \textsc{Trivance}, the latency-optimal algorithm induces congestion of $3^k$ for $\log_3 n$ steps which converges via geometric series to $\frac{n}{2}$. In the bandwidth-optimal variant, the message size decays as $\frac{1}{3^{k+1}}$ while congestion remains $3^k$, resulting in a constant per step product of $\frac{1}{3}$. Over $2 \cdot \log_3 n$ steps, this equals an optimality of $\frac{2}{3} \cdot \log_3 n$. Bruck's data transmission follows the same pattern as \textsc{Trivance}, but in step $k$ the communication distances are $3^k$ and $2\cdot 3^k$, resulting in a total congestion of $3\cdot 3^k$ per step. Consequently, the transmission delay of Bruck's algorithm is exactly three times that of \textsc{Trivance}, for both the latency-optimal and the bandwidth-optimal variants.
\subsection{For Multidimensional Torus}
The precise transmission delay optimality for the bandwidth-optimal algorithms in multidimensional torus networks are:
\begin{itemize}
  \item Recursive Doubling $\sum_{d=0}^{D-1} \sum_{k=0}^{\frac{\log_2 n}{D}-1} (\frac{1}{2^{1+d+D\cdot k}} \cdot 2^k)$
  \item Swing $\sum_{d=0}^{D-1} \sum_{k=0}^{\frac{\log_2 n}{D}-1} \frac{1}{2^{\,1+d+D\cdot k}} \cdot  \tfrac{2^{k+1}-(-1)^{k+1}}{3}$
  \item \textsc{Trivance} $2\cdot \sum_{d=0}^{D-1} \sum_{k=0}^{\frac{\log_3 n}{D}-1} (\frac{1}{3^{1+d+D\cdot k}} \cdot 3^k)$
  \item Bruck $2\cdot \sum_{d=0}^{D-1} \sum_{k=0}^{\frac{\log_3 n}{D}-1} (\frac{1}{3^{1+d+D\cdot k}} \cdot 3 \cdot 3^k)$
\end{itemize}
For each term $\sum_{d=0}^{D-1}$ represents the iteration over available dimensions, while for instance for Recursive Doubling $\sum_{k=0}^{\frac{\log_2 n}{D}-1}$ represents the number of steps required to perform in this said dimension. The data size is reduced for each step regardless of performed in the same or different dimension.

For latency-optimal approaches, we can calculate the cost of communication in a single dimension and multiply it with the number of dimensions, it follows:
\begin{itemize}
  \item Recursive Doubling $D\cdot\sum_{k=0}^{\frac{\log_2 n}{D}-1} 2^k$
  \item Swing $D\cdot\sum_{k=0}^{\frac{\log_2 n}{D}-1} \lceil\tfrac{2^{k+1}-(-1)^{k+1}}{6}\rceil$
  \item \textsc{Trivance} $D\cdot\frac{1}{D}\sum_{k=0}^{\frac{\log_3 n}{D}-1} 3^k$
  \item Bruck $D\cdot\frac{1}{D}\sum_{k=0}^{\frac{\log_3 n}{D}-1}  3 \cdot 3^k)$
\end{itemize}
\textsc{Trivance}\ and Bruck perform one collective operation per dimension. In contrast, the latency-optimal variants of Recursive Doubling and Swing utilize only a single port per node~\cite{295653}.

\section{Bruck Implementation Details}
\label{bruckimplementation}

Bruck's communication pattern was originally designed for the All-to-All or AllGather collective, where data blocks are forwarded without intermediate reduction on a logical topology. For Bruck's AllGather, the number of transmitted blocks grows exponentially across steps: one block in the first step, then two, four, and so on. When it is deployed on a physical ring, the communication distance also increases exponentially. Paired together, the transmission delay cost increases exponentially. Following Jeaugey~\cite{jeaugey2025pat}, we therefore reverse the AllGather schedule: single or few blocks are sent over the largest distances first, while multiple blocks are sent later over shorter distances. This preserves the logical Bruck data propagation, but substantially reduces transmission delay on a physical ring. The set of blocks transmitted in each step is computed using the same block-propagation procedure as for Recursive Doubling, Swing, and \textsc{Trivance}, following Algorithm~\ref{alg:get_blocks}. Finally, we map Bruck's logical communication pattern to the physical ring using shortest-path routing. In the two-port radix-three variant, a step may address logical offsets $3^k$ and $2\cdot 3^k$. For large $k$, these offsets can exceed the ring diameter when interpreted in a fixed logical direction. We therefore do not force Bruck traffic to follow a single direction around the ring; instead, each logical exchange is routed along the shorter physical path. All remaining baselines, namely Recursive Doubling, Bucket, and Swing, are already designed for AllReduce on ring or torus topologies~\cite{10.1145/2686882, 10.1145/1810085.1810093, 10.1007/978-3-540-24685-5_1, 295653}; we evaluate them by using all available $2D$ ports.

\section{Integration into Collective Libraries}
\label{integrationNCCL}

Integrating \textsc{Trivance} does not require changes to the user-facing collective API, but requires a new backend schedule. Following existing Ring and Bruck-derived PAT backends~\cite{jeaugey2025pat}, the scheduler emits one work descriptor per rank, step, and chunk, containing both peer connections, buffer offsets, the active block range, and the element count. PAT operations have one receive dependency, whereas a \textsc{Trivance} Reduce-Scatter operation binds two incoming chunks to the same output range. The AllGather phase uses the same communication schedule but does not require a joint reduction.

The two receive dependency can be implemented through a per-chunk readiness entry indexed by the logical step, block range, and chunk index. Completion of the left and right receives sets one of two readiness bits; the reduction becomes executable only when both bits are set and the entry belongs to the current buffer generation. The reduction primitive directly reads the two incoming chunks from their receive buffers and combines them with the local partial result, producing the partial result required for the subsequent step without first concatenating the incoming chunks. The readiness entry and associated receive slots remain valid until the reduction completes. This design preserves chunk-level overlap while introducing a local two receive readiness condition into the progress engine.

The main deployment challenge is topology-aware connection mapping. A logical communication channel represents a scheduling stream but does not necessarily correspond to a distinct physical path. For each rank and torus dimension, the topology discovery stage must therefore expose the peers in the positive and negative directions together with the NICs, ports, and routes available to reach them. The connection mapper must bind the two peer connections to distinct resources that preserve the intended torus directions and support concurrent transmission. Mappings that serialize the two streams through a shared egress port, NIC bottleneck, or overlapping physical route must be excluded. This differs from the PAT algorithm, whose generated operations contain only one receive peer and one send peer and therefore do not require two simultaneously usable paths. A mapping with overlapping resources preserves the functional correctness of \textsc{Trivance} but may degrade its concurrency and congestion. The algorithm selection policy should therefore select \textsc{Trivance} only when the topology provides two suitable paths for each active dimension and otherwise select an existing collective algorithm.

\section{Simulation Parameters}
Table~\ref{tab:sst-parameters} summarizes the most relevant SST parameters used throughout the evaluation which are based on the simulation setup from De Sensi et al.~\cite{295653}. Unless stated otherwise, all experiments use the same configuration. The simulated torus uses deterministic shortest-path routing for all algorithms. Each transmission traverses links within a single torus dimension, and the collective schedule determines the communication direction. For \textsc{Trivance}, all peer distances are strictly smaller than half the corresponding dimension size, so the selected shortest path is unique. Equal-distance ties in other algorithms are resolved deterministically by their schedule. Consequently, no alternative minimal path remains over which adaptive routing could load-balance traffic in our experiments. We therefore do not vary the routing policy. SST/Merlin models network congestion at packet level. Endpoints connect through \texttt{merlin.linkcontrol}, which uses credit-based flow control. Credits are initialized from the configured input and output buffer capacities. A packet can be injected only if sufficient downstream credits are available; otherwise, injection stalls until credits are returned as receive queues are drained. Congestion is therefore captured through queuing, flow-control stalls, endpoint injection limits and link congestion.

Each physical link provides a bandwidth of 800\,Gb/s and 100\,ns propagation delay. Routers introduce 50\,ns of input and output pipeline latency, resulting in a per hop delay of approximately 200\,ns. Messages are segmented into 8192\,B packets. We additionally evaluated link bandwidths ranging from 200\,Gb/s to 3200\,Gb/s in Figure~\ref{fig:bandwidth}, corresponding to current and projected datacenter interconnect speeds. Each node contains two ports per torus dimension, enabling simultaneous bidirectional communication along every dimension. For a $d$-dimensional torus, this results in $2d$ active network interfaces per node. The Firefly/Hermes endpoint stack models startup latency and includes software entry and return overheads, message matching, transmission setup costs, host--NIC interaction latency, and NIC transmission delay. Each interface has a dedicated send/receive engine and injects traffic into its corresponding link independently. The reported link bandwidth denotes the transmission rate supported by each directional connection and determines the per-byte transmission cost. We do not impose an additional node-wide injection limit shared across interfaces. This models direct-connect accelerator interconnects such as TPUv4’s six dedicated ICI links~\cite{TPUv4}. The separate NIC objects used by SST represent these interfaces rather than distinct physical NIC cards.

\begin{table}[H]
\centering
\small
\begin{tabular}{ll}
\toprule
\textbf{Parameter} & \textbf{Value} \\
\midrule
Packet size & 8192\,B \\
Link bandwidth & 800\,Gb/s \\
Link propagation delay & 100\,ns \\
Router input latency & 50\,ns \\
Router output latency & 50\,ns \\
Input buffer size & 64\,MB \\
Output buffer size & 64\,MB \\
NIC interfaces per node & 2D \\
\bottomrule
\end{tabular}
\caption{Key SST/Merlin configuration parameters.}
\label{tab:sst-parameters}

\end{table}

\section{Sensitivity Analysis}
\label{app:sensitivity}

We evaluate the robustness of the observed performance trends to changes in key network and simulation parameters. Increasing the startup latency in Figure~\ref{fig:sens_8x8_step} shifts the tradeoff point toward larger AllReduce sizes because \textsc{Trivance} requires fewer logical communication steps and avoids the factor-$\log_2 3$ step overhead of $\log_2 n$-step schedules. Increasing the link latency in Figure~\ref{fig:sens_8x8_link} produces a similar but weaker effect: \textsc{Trivance} reduces the aggregate communication path length compared with existing logarithmic algorithms, but this advantage affects only the per hop propagation cost. Figure~\ref{fig:sens_8x8_packet} shows that varying the packet size from 1\,KiB to 64\,KiB has little effect on the relative performance, with all algorithms exhibiting the same qualitative trend across the evaluated range. Bandwidth sensitivity is discussed separately in Section~\ref{eva:BW}; higher bandwidth shifts the tradeoff point toward larger AllReduce sizes by reducing the relative contribution of data transmission time.

\begin{figure}[b]
    \centering
        \includegraphics[width=\columnwidth]{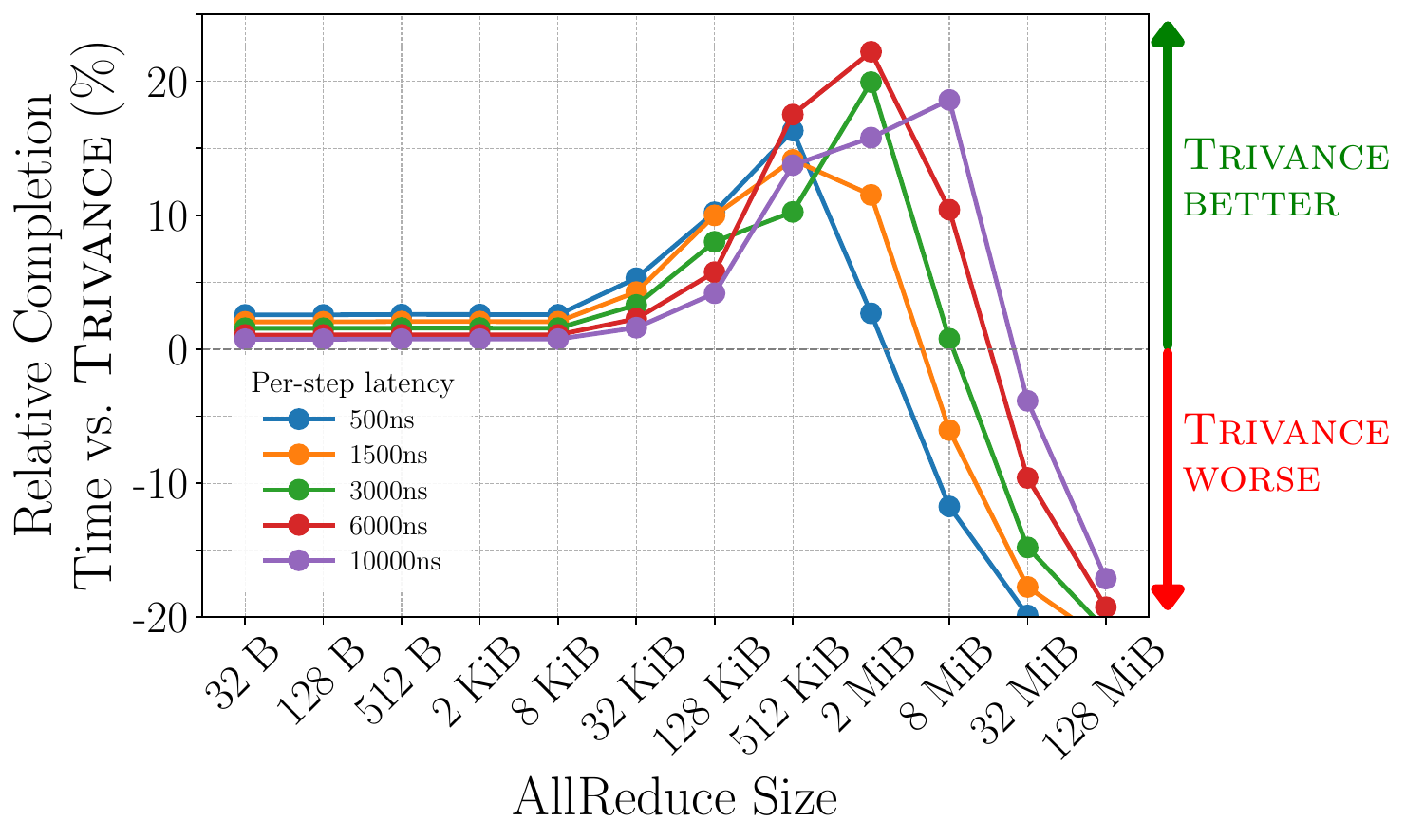}
                        \Description{$8\times8$ torus - per step latency comparison}
  \caption{$8\times8$ torus of \textsc{Trivance} compared to the best performing baseline for various per step latencies.}
  \label{fig:sens_8x8_step}
\end{figure}

\begin{figure}[b]
    \centering
        \includegraphics[width=\columnwidth]{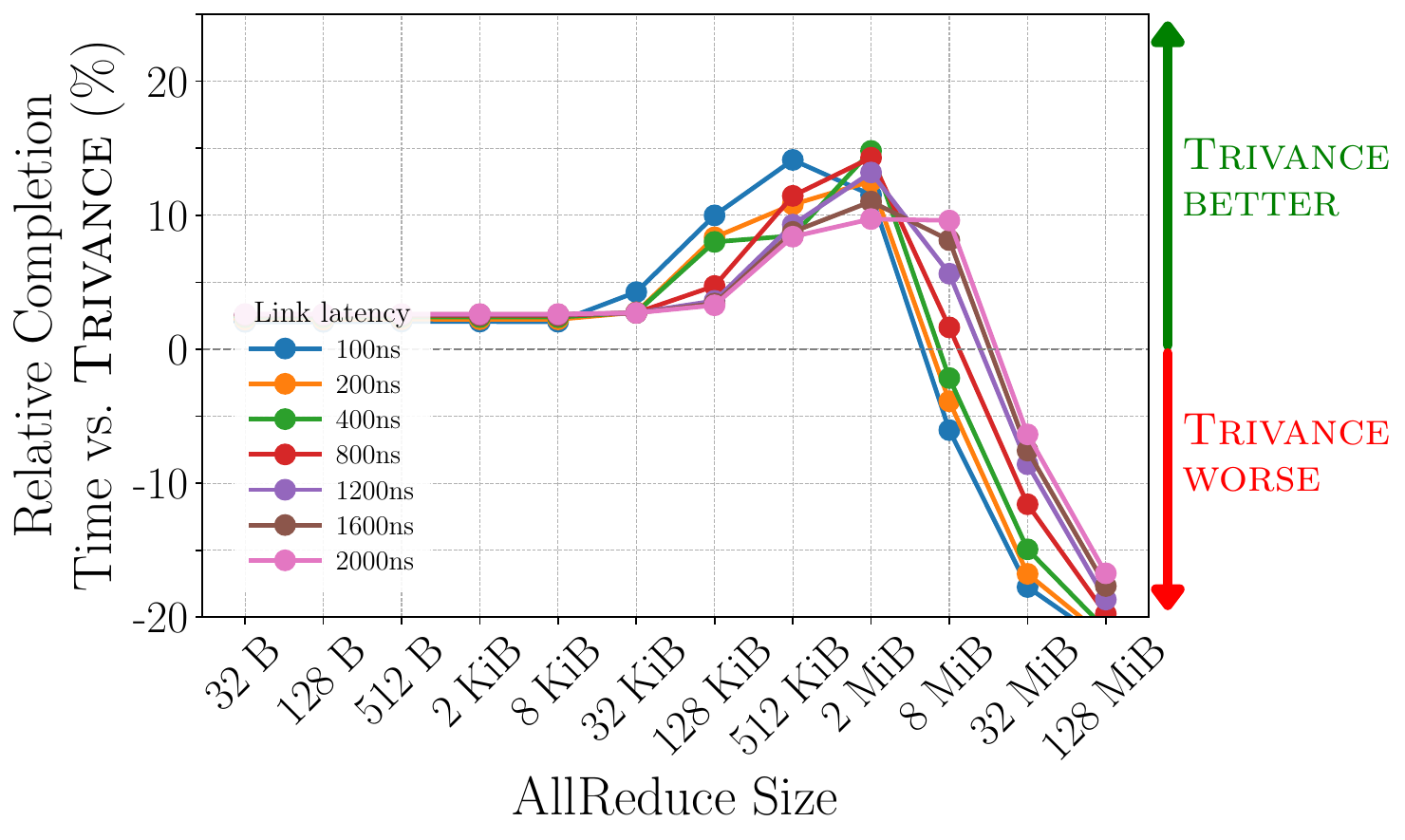}
                        \Description{$8\times8$ torus - link latency comparison}
  \caption{$8\times8$ torus of \textsc{Trivance} compared to the best performing baseline for various link latencies.}
  \label{fig:sens_8x8_link}
\end{figure}

\begin{figure}[t]
    \centering
        \includegraphics[width=\columnwidth]{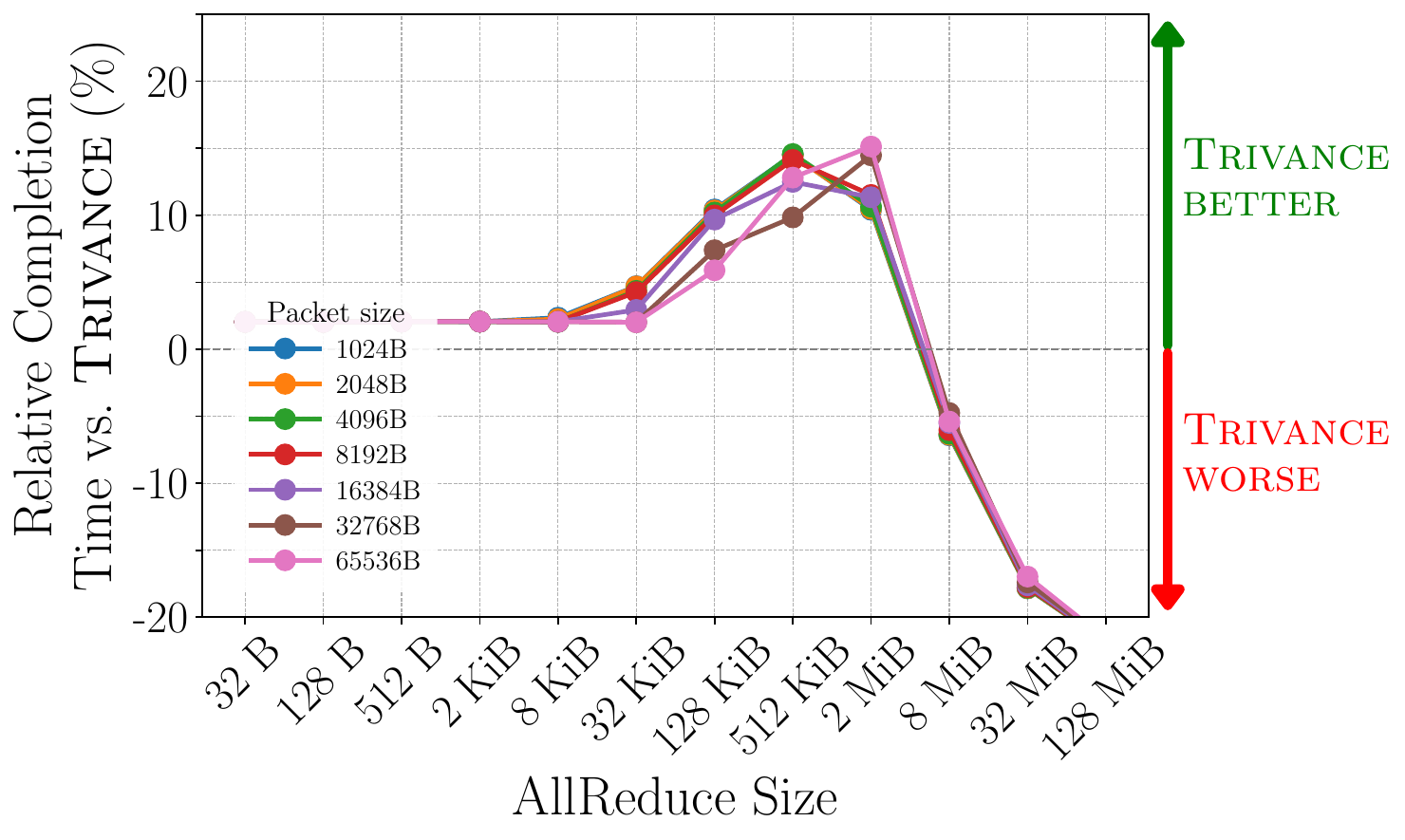}
                \Description{$8\times8$ torus - packet size comparison}
  \caption{$8\times8$ torus of \textsc{Trivance} compared to the best performing baseline for various packet sizes.}
  \label{fig:sens_8x8_packet}
\end{figure}

\section{Additional Plots}
\label{appendix:add}
\begin{figure}[H]
        \centering
        \includegraphics[width=\columnwidth]{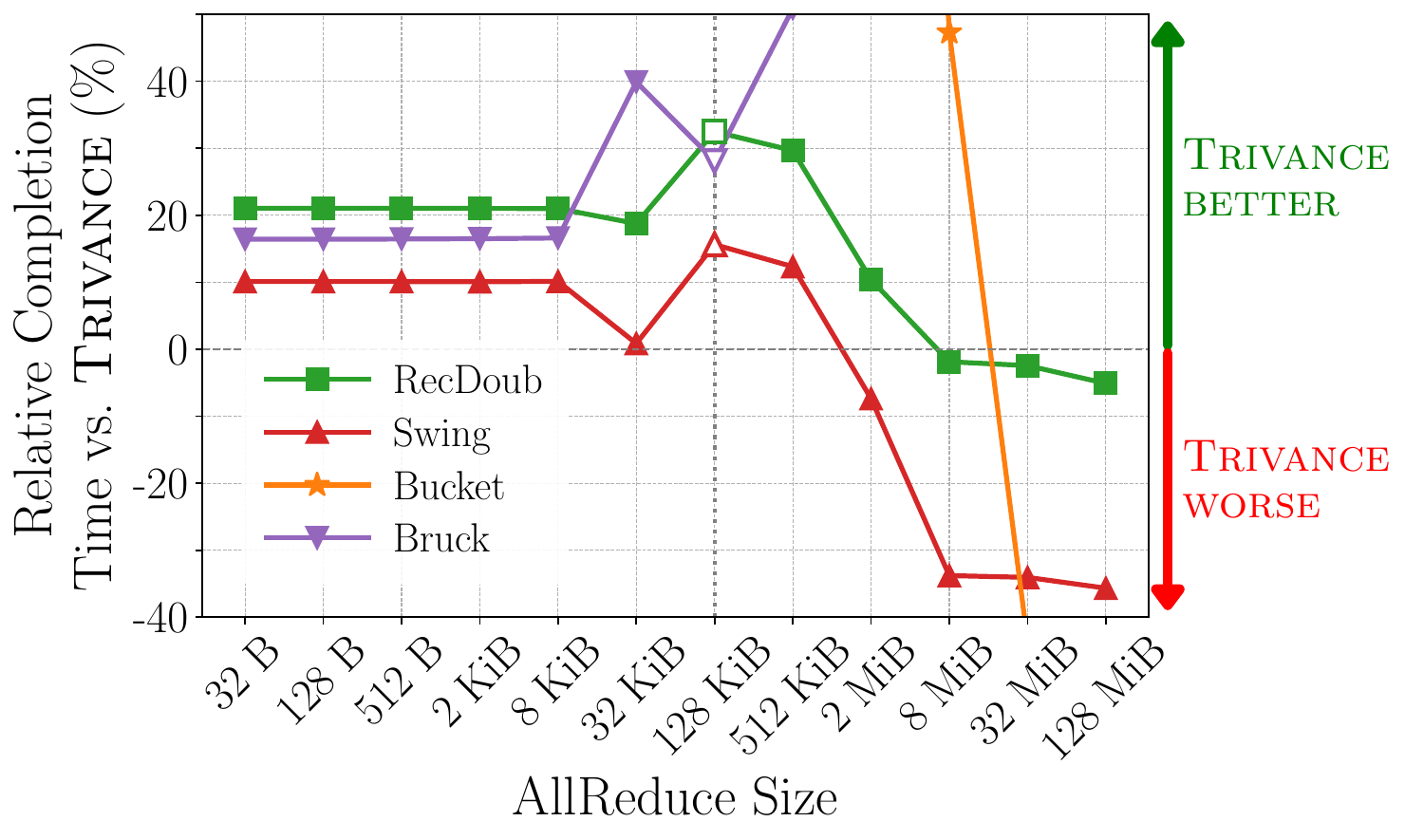}
                \Description{$128$ torus}
        \caption{128 node ring}
        \label{fig:torus128}
\end{figure}

\begin{figure}[H]
        \centering
        \includegraphics[width=\columnwidth]{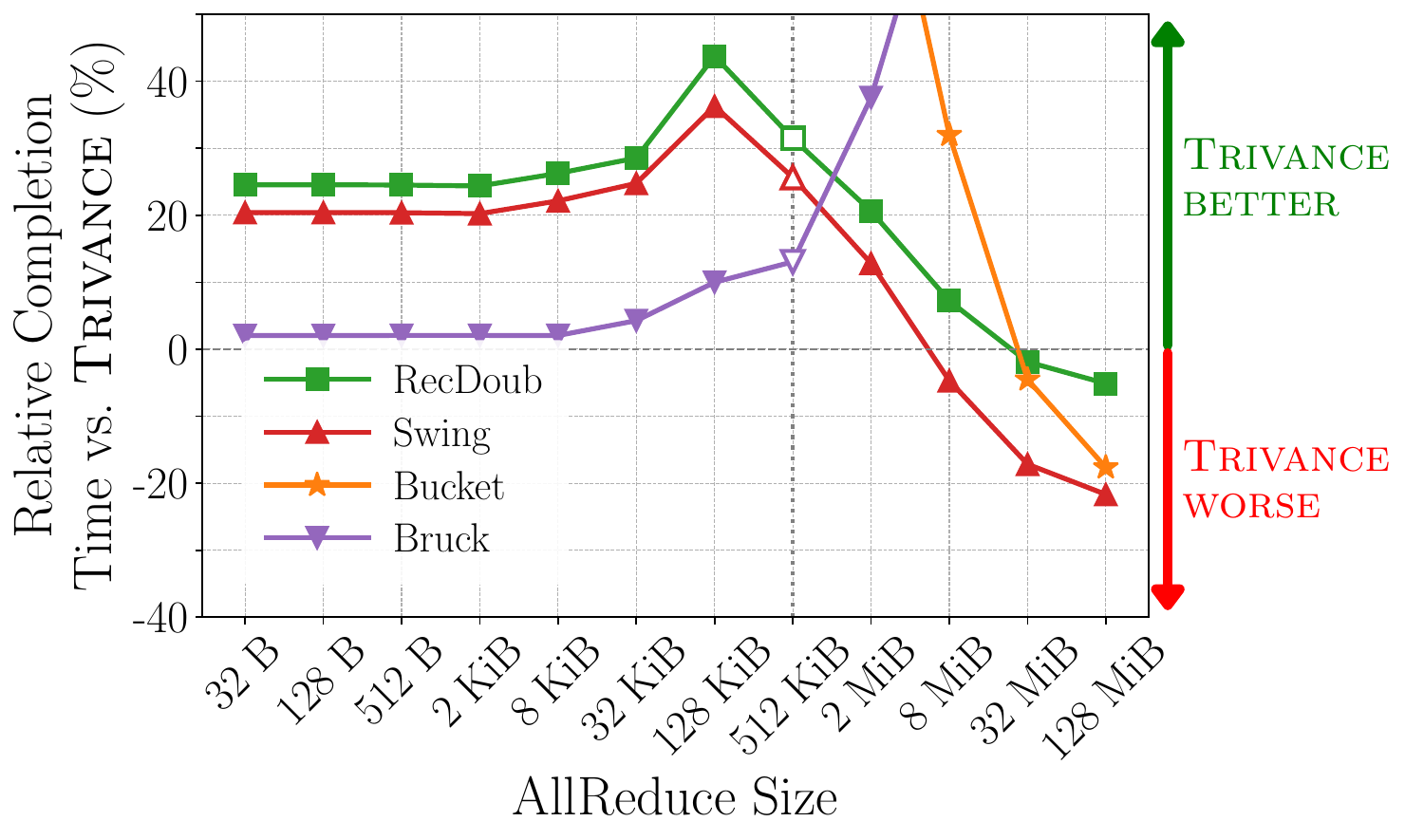}
                \Description{$8\times8$ torus}
        \caption{$8\times8$ torus}
        \label{fig:torus8x8}
\end{figure}

\begin{figure}[H]
        \centering
        \includegraphics[width=\columnwidth]{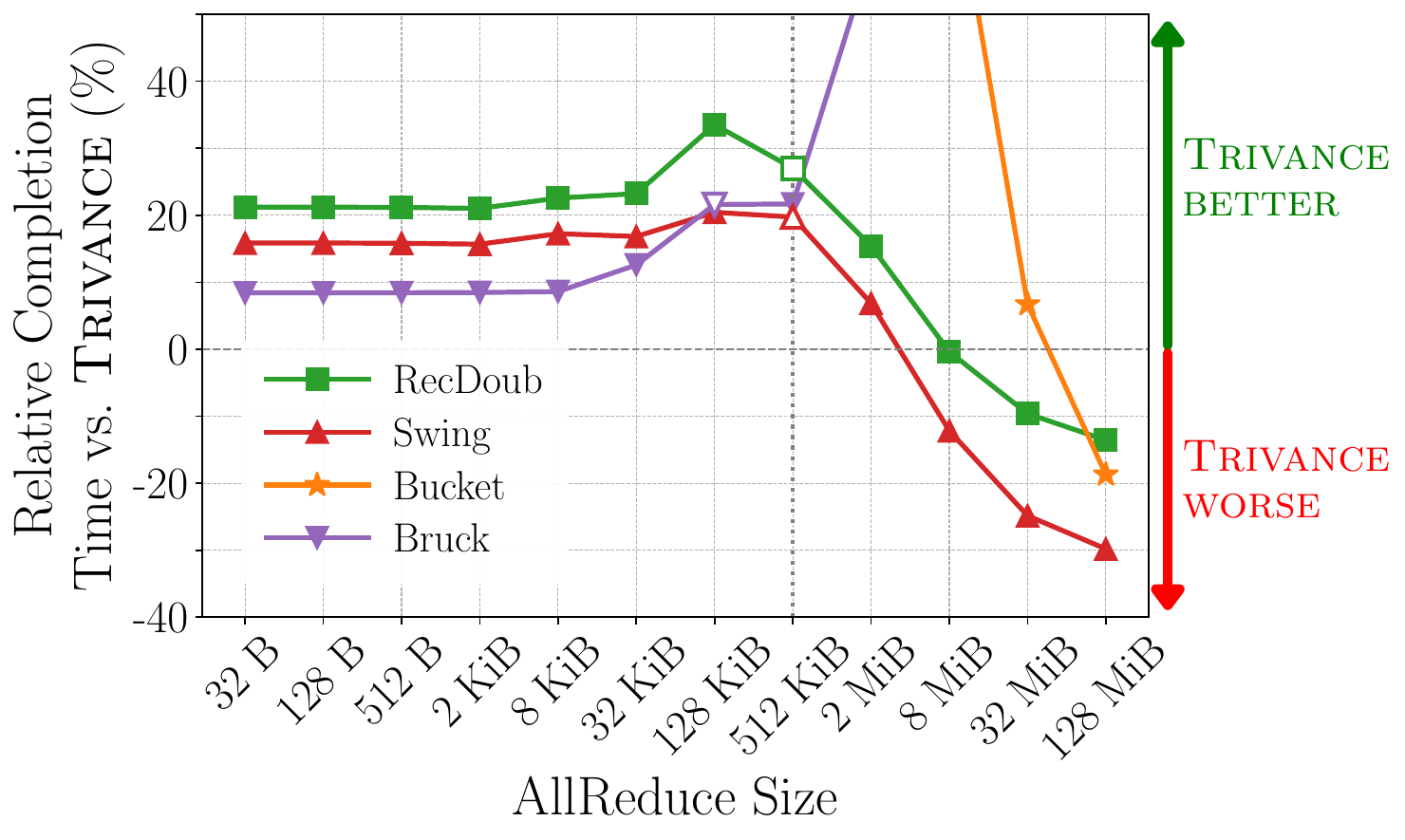}
                \Description{$8\times16$ torus}
        \caption{$8\times16$ torus}
        \label{fig:torus8x16}
\end{figure}
\vspace{54px}
\begin{figure}[H]
        \centering
        \includegraphics[width=\columnwidth]{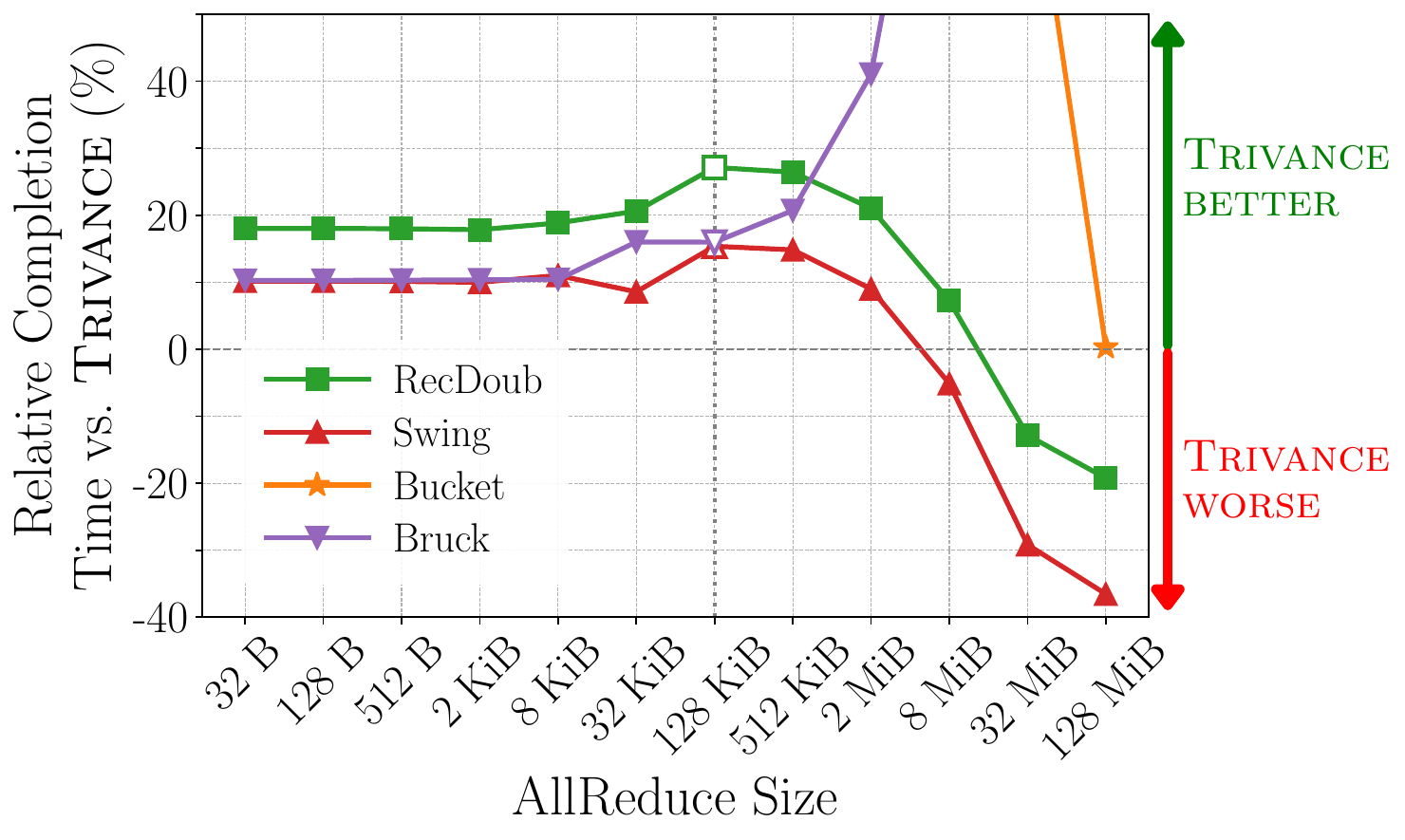}
                \Description{$16\times64$ torus}
        \caption{$16\times64$ torus}
        \label{fig:torus16x64}
\end{figure}
\vspace{20px}
\begin{figure}[H]
        \centering
        \includegraphics[width=\columnwidth]{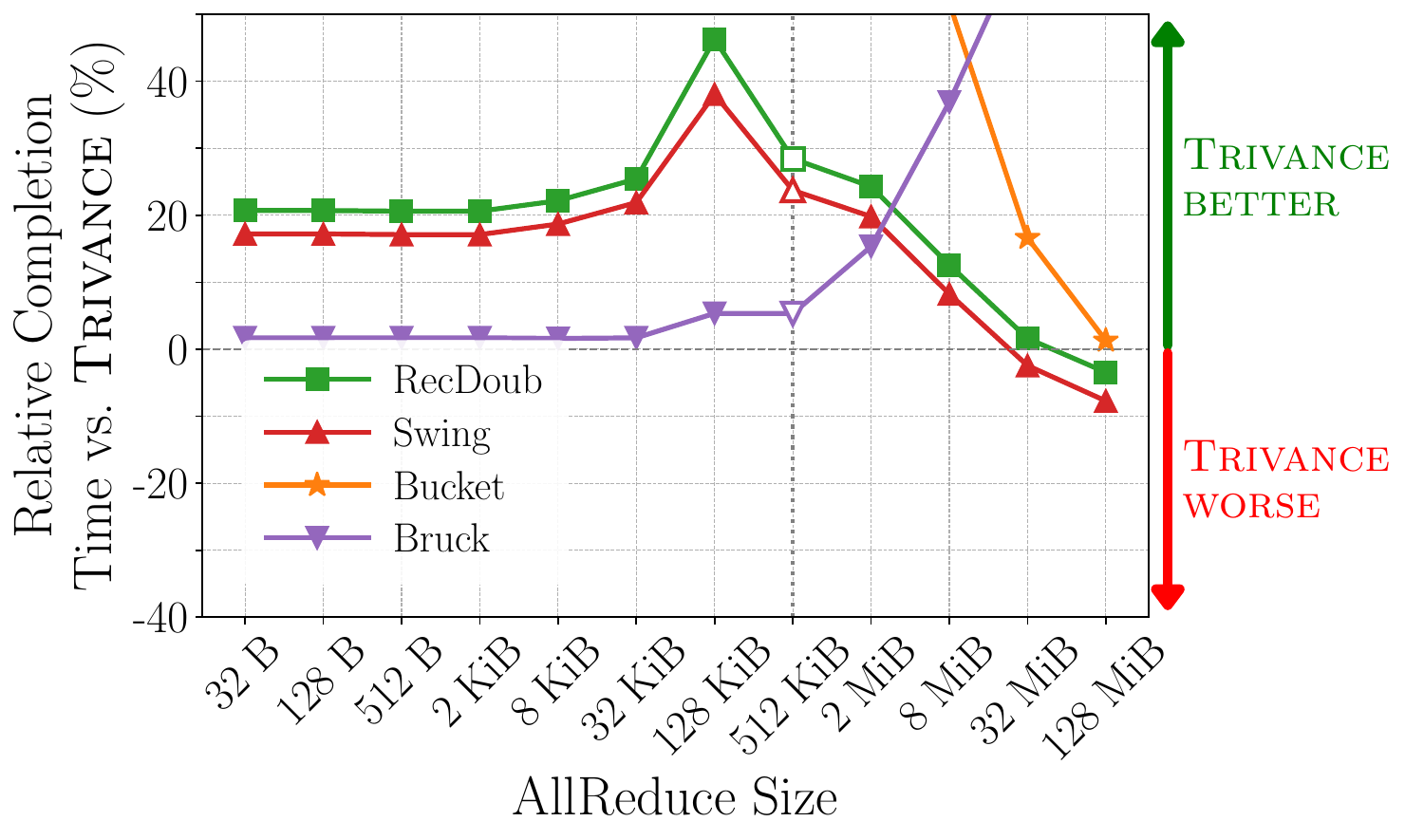}
        \Description{$8\times8\times8$ torus}
        \caption{$8\times8\times8$ torus}
        \label{fig:torus8x8x8}
\end{figure}

\label{LastPage}
\end{document}